\documentclass[11pt,a4paper]{article}
\usepackage[english]{babel}
\usepackage[latin1]{inputenc}
\usepackage{amsmath, amsthm, amssymb}
\usepackage[pdftex]{graphicx}
\usepackage{amssymb}
\usepackage{latexsym}
\usepackage{amstext}
\usepackage{epstopdf}
\usepackage{floatflt}
\usepackage{hyperref}
\usepackage{enumitem}
\usepackage{multirow}
\usepackage{pdflscape}
\usepackage{subfigure}
\usepackage{tikz}
\usetikzlibrary{positioning,arrows}
\usetikzlibrary{decorations.pathmorphing}
\usetikzlibrary{decorations.markings}
\usetikzlibrary{3d}

\newtheorem{theorem}{Theorem}[section]

\newtheorem{definition}[theorem]{Definition}

\newtheorem{proposition}[theorem]{Proposition}
\newtheorem{remark}[theorem]{Remark}

\newcommand{\aaa}{\mathfrak{a}}
\newcommand{\ad}{\aaa_2}

\newcommand{\adc}{\aaa_2^c}

\newcommand{\adf}{\aaa_2^f}
\newcommand{\adgd}{\aaa_2^{(2)}}
\newcommand{\adgu}{\aaa_2^{(1)}}

\newcommand{\as}{\alpha}

\newcommand{\au}{\mathbf a_1}

\newcommand{\ba}{\begin{array}}

\newcommand{\bb}{\beta}
\newcommand{\bc}{\begin{center}}
\newcommand{\bdo}{{\mathfrak B}_{12}}
\newcommand{\be}{\begin{equation}}
\newcommand{\bea}{\begin{equation}\begin{array}}
\newcommand{\beas}{\begin{equation*}\begin{array}}

\newcommand{\bef}{\begin{flalign}}
\newcommand{\befs}{\begin{flalign*}}
\newcommand{\ben}{\begin{enumerate}}
\newcommand{\benal}[1]{\begin{enumerate}[label={#1}\alph*)]}
\newcommand{\benar}[1]{\begin{enumerate}[label={#1}\arabic*)]}
\newcommand{\benro}[1]{\begin{enumerate}[label={#1}\roman*)]}
\newcommand{\benRo}[1]{\begin{enumerate}[label={#1}\Roman*)]}
\newcommand{\bes}{\begin{equation*}}
\newcommand{\bfit}[1]{\bf\emph{#1}}
\newcommand{\bit}{\begin{itemize}}

\newcommand{\borc}{\mathfrak B}
\newcommand{\Bp}{{\mathfrak B}^+}

\newcommand{\cc}{\mathbb{C}}

\newcommand{\dedo}{\mathbf{de_{12}}}

\newcommand{\dotto}{\mathbf{d_8}}

\newcommand{\ea}{\end{array}}
\newcommand{\ec}{\end{center}}
\newcommand{\ed}{\mathbf{e_{10}}}
\newcommand{\edo}{\mathbf{e_{12}}}

\newcommand{\ee}{\end{equation}}
\newcommand{\eea}{\end{array}\end{equation}}
\newcommand{\eeas}{\end{array}\end{equation*}}
\newcommand{\eef}{\end{flalign}}
\newcommand{\eefs}{\end{flalign*}}
\newcommand{\een}{\end{enumerate}}
\newcommand{\ees}{\end{equation*}}

\newcommand{\eit}{\end{itemize}}

\newcommand{\eno}{\mathbf{e_9}}
\newcommand{\eo}{\mathbf{e_8}}

\newcommand{\ep}{\varepsilon}
\newcommand{\eps}{\varepsilon}

\newcommand{\es}{\mathbf{e_6}}

\newcommand{\fieo}{\Phi_8}
\newcommand{\fio}{\Phi_{\cal O}}
\newcommand{\fis}{\Phi_{\cal S}}

\newcommand{\ga}{{\mathfrak g}(A)}
\newcommand{\gas}{{\mathfrak g}_\as}

\newcommand{\ggo}{\mathfrak{g}_0}
\newcommand{\ggu}{\textbf{${\mathfrak g}_{\mathsf u}$}}

\newcommand{\gh}{\gamma}

\newcommand{\hei}{\text{ht}}

\newcommand{\hh}{\mathfrak{h}}
\newcommand{\hhs}{{\mathfrak h}^*}

\newcommand{\ical}{\mathcal I}
\newcommand{\impl}{\ \Rightarrow\ }
\newcommand{\iinu}{\text{II}_{9,1}}
\newcommand{\imi}{\mathbf{i}\, }
\newcommand{\iu}{\imi}

\newcommand{\km}{k_{\text{-}1}}
\newcommand{\kop}{k_{0'}}
\newcommand{\kopp}{k_{0''}}
\newcommand{\lam}{\lambda}
\newcommand{\Ll}{\mathbb L}
\newcommand{\Lleo}{{\mathbb L}_{\eo}}

\newcommand{\ncm}{\nu^\prime_\chi}
\newcommand{\ncp}{\bar\nu^\prime_\chi}
\newcommand{\nem}{\nu^\prime_e}
\newcommand{\nep}{\bar\nu^\prime_e}

\newcommand{\nn}{\mathbb N}
\newcommand{\nm}{{\mathfrak n}_-}
\newcommand{\nmm}{\nu^\prime_\mu}
\newcommand{\nmp}{\bar\nu^\prime_\mu}
\newcommand{\np}{{\mathfrak n}_+}

\newcommand{\ntm}{\nu^\prime_\tau}
\newcommand{\ntp}{\bar\nu^\prime_\tau}

\newcommand{\php}[1]{\phantom{#1}}

\newcommand{\pt}{\tilde p}

\newcommand{\rmu}{\as_{\text{-}1}}

\newcommand{\rop}{\as_{0'}}
\newcommand{\ropp}{\as_{0''}}

\newcommand{\rr}{\mathbb{R}}

\newcommand{\sdel}{\tilde\delta}

\newcommand{\spin}{\mathbf{su(2)^{spin}}}

\newcommand{\sref}[1]{{\bf\ref{#1}}}

\newcommand{\sud}{\mathbf{su(2)}}

\newcommand{\um}{{\scriptstyle \frac12}}

\newcommand{\zz}{\mathbb Z}

\numberwithin{equation}{section}
\begin{document}

\begin{titlepage}

\vskip 2.0 cm
\begin{center}  {\Huge{\bf Space, Matter and Interactions\\\vskip 0.2 cm in a Quantum Early Universe}} \\\vskip 0.5 cm {\huge{\bf Part I : Kac-Moody and Borcherds Algebras}}

\vskip 2.5 cm

{\Large{\bf Piero Truini$^{1,2}$},  {\bf Alessio Marrani$^{3,4}$},\\ {\bf Michael Rios$^{2}$}, {\bf Klee Irwin$^{2}$}}

\vskip 1.0 cm

$^1${\sl INFN, sez. di Genova, via Dodecaneso 33, I-16146 Genova, Italy\\
	\texttt{piero.truini@ge.infn.it}}\\

\vskip 0.5
cm

$^2${\sl
	QGR, 101 S. Topanga Canyon Rd., 1159 Los Angeles,CA 90290, USA}\\

\vskip 0.5 cm

$^3${\sl Centro Studi e Ricerche Enrico Fermi, via Panisperna 89A,
I-00184, Roma, Italy}\\

\vskip 0.5 cm

$^4${\sl Dipartimento di Fisica
e Astronomia Galileo Galilei, Universit\`a di Padova,\\and INFN, sezione
di
Padova, Via Marzolo 8, I-35131 Padova, Italy\\
\texttt{jazzphyzz@gmail.com}}\\

\vskip 0.5 cm

\vskip 0.5 cm

 \end{center}

 \vskip 4.0 cm

\begin{%
abstract}

We introduce a quantum model for the Universe at its early stages, formulating a mechanism for the expansion of space and matter from a quantum initial condition, with particle interactions and creation driven by algebraic extensions of the Kac-Moody Lie algebra $\eno$. We investigate Kac-Moody and Borcherds algebras, and we propose a generalization that meets further requirements that we regard as fundamental in quantum gravity.

 \end{abstract}
\vspace{24pt} \end{titlepage}


\newpage \tableofcontents \newpage


\section{Introduction}

This is the first of two papers -see also \cite{mym2}- describing an
algebraic model of quantum gravity.\newline
We believe it is a fresh and original attempt to discuss quantum gravity,
based only on the fundamental principles of quantum physics and on well
tested experimental evidence, without being biased by any other model or
theory. In this first paper, we provide an introduction to the foundations
of the model, and we start investigating the mathematical structures that
suit our purpose. In the second paper, \cite{mym2}, we will deal with a
physical model relying on a particular infinite dimensional algebra.\newline

The Lie algebra at the core of our model has the following features and
interpretation:

\begin{enumerate}
\item it is an infinite dimensional Lie algebra extending $\eo$, that is
regarded as the internal-quantum-number subalgebra, meaning that the $\eo$
roots represent charges and spin of elementary particles;

\item its root lattice is Lorentzian;

\item the subspace of the lattice complementary to that of $\eo$ is
interpreted as momentum space.
\end{enumerate}

\begin{remark}\label{rdg}
	The Lie algebra $\eo$ has been considered by many as a possible algebra for grand-unification as well as for quantum gravity. It has then been considered not suitable after the no-go theorem by Distler and Garibaldi, \cite{dg}. We will show in section \sref{stdm} how $\eo$ fulfills the requirements for standard model degrees of freedom and algebras, which seems to contradict the thesis of Distler and Garibaldi. We underline here that it does not, since the hypothesis denoted TOE1 by the authors of \cite{dg}, in particular the fact that the algebra of the standard model centralizes $\mathbf{sl(2,\cc)}$, not only does not apply, but actually needs not to do so, as it will become obvious in the development of section \sref{stdm}.
\end{remark}

Algebraic methods are extensively used and successfully exploited in string
theory and conformal field theory in two dimensions, through the concept of
vertex operator algebras, \cite{belavin}-\nocite{verl}\cite{huang}, in order
to describe the interactions between different strings, localized at
vertices, analogously to the Feynman diagram vertices. Mathematically, the
underlying concept of a vertex algebra was introduced by Borcherds, \cite%
{borch}-\nocite{frenk}\cite{kacvoa}, in order to prove the \textit{Monstrous
Moonshine} conjecture, \cite{conway1}.

The algebras used in this paper may be regarded as \textit{vertex operator
algebras} in a broader sense, since they are characterized by interaction
operators that look like generators of a Lie algebra, and whose product
depends upon parameters related to the spacetime creation and expansion. The
Lie algebra acts \textit{locally}, but it is immersed in a wider,
vertex-type algebra by means of a mechanism which creates a discrete quantum
spacetime.

The Pauli exclusion principle is fulfilled by turning the algebra into a Lie
superalgebra, using the \textit{Grassmann envelope}.

The resulting model is thus intrinsically relativistic, both because of the
way spacetime expands and because the Poincar\'{e} group acts locally on the
Lie algebra. Furthermore, the conservation of charge and momentum is a
consequence of the Lie product, and in this respect they are treated at the
same level.

\subsection{{\Large $\ggu$}, a Lie Algebra for Quantum Gravity}

At a very fundamental level, we make the following assumptions on quantum
gravity, founded on the current theoretical and experimental knowledge in
physics.

\benar{QG.}
\item gravity is a characteristic of spacetime;

\item spacetime is \textit{dynamical} and related to matter.
Therefore, we assume that it emerges from the existence of particles and
their interactions. There is no way of defining distances and time lapses
without interactions, so that the creation and expansion of spacetime is
itself a rule followed by particle interactions;

\item a suitable mathematical structure at the core of the
description of quantum gravity is that of an algebra, which we will
henceforth denote by $\ggu$, whose generators represent the particles and
whose product yields the building blocks of the interactions (let us call
them \textit{elementary interactions}). As a consequence, the interactions
are endowed with a tree structure, thus opening the way for a description of
scattering amplitudes in terms of what we would call \textit{gravitahedra},
providing a generalization of the associahedra and permutahedra in the
current theory of scattering, \cite{marni}-\nocite{nima1}\nocite%
{miz}\nocite{nima2}\nocite{stash1}\nocite{stash2}\nocite{tonks}\cite{loday}.
The structure constants of the algebra determine the quantum amplitudes of
the elementary interactions; in particular, we assume $\ggu$ to be a Lie
algebra, because it enables to derive the fundamental conservation laws
observed in physics directly from the action of the generators as
derivations (Jacobi identity); as in the theory of fields, the interactions
may only occur locally, point by point in the expanding spacetime, that can
therefore be viewed as a parameter on which the algebra product depends;

\item in agreement with the theory of a big bang, strongly supported
by the current observations, we assume the existence of an initial quantum
state, mathematically represented by an element of the universal enveloping
algebra of $\ggu$. Such an element is made of generators that can all
interact among themselves, thus yielding the first geometrical
interpretation: that of a point where particles may interact;

\item a particle has a certain probability amplitude to interact but
also not to interact, in which case it expands as described in section \ref%
{shift};

\item particles are quantum objects, hence their existence through
interactions occurs with certain amplitudes. Therefore, \textbf{spacetime
acquires} \textbf{a quantum structure}: a point in space and time is where
particles are present with a certain amplitude and may interact. The
amplitude related to the quantum spacetime point is the sum of the
amplitudes for particles to be there. Consequently, the fact that
gravitation appears as an attractive force has to be explained through
amplitudes and their interference;

\item there is a universal clock, which provides the ordering of the
finite and discrete number of interactions. The expansions are also
countable, hence discrete: the structure of spacetime that emerges is
discrete and finite, as is the Universe and the quantum theory describing
it. There is no divergence of any sort: quantum field theory in the
continuum, with its divergences and related renormalizations, is an
approximation that may be useful for calculations long after the big bang;

\item the finiteness of the expanding Universe, and thus the absence
of spacetime beyond it, affects the quantum initial state of particles,
which are not free to move on the spacetime \textit{stage} but are bound as
if they were surrounded by infinitely high barriers. The steady state of
such a particle is a \textbf{superposition of states with} \textbf{opposite
3-momenta}, representing an object that moves simultaneously in opposite
directions, where by 3-momentum we denote the spatial component of
4-momentum.
\een

\subsection{Expansion\label{shift}}

The assignment of opposite 3-momenta is inherent to the quantum behavior of
a particle in a box, in which the square of the momentum, but not the
momentum itself, has a definite value in a stationary state. In standard
relativistic and non-relativistic quantum mechanics, the ground state is a
superposition of \textit{generalized states} with opposite momenta $e^{\iu%
(kx+a)}$ and $e^{-\iu(kx+a)}$. We maintain the same energy and start
enlarging the box on opposite sides along the direction of $\vec{p}$ in
steps of $|\vec{p}|/E$ in Planck units, so that a massless particle travels
at velocity $c$ and a massive one slower than that. We get a wave
proportional to $\sin (\dfrac{\pi }{2}(x/a+n))$, for $a=1=|\vec{p}|/E$ and $%
n=1,2,3,...$.

The wave function for the first four expansions is shown in Fig. \ref%
{step5B}. We take the discretized picture of the sine function maxima and
minima, the dots, with cosmological time $t=n-1$.

\begin{figure}[hbtp]
	\centering%
	\begin{tikzpicture}
	\draw[thick,->] (-1.2,1) -- node[above]{$\vec p/E$} (-2.2,1);
	\draw[->] (-5,0) -- node[near end, right]{amp.} (-5,1.5);
	\draw (-6,0) node[above]{$t=0$} -- (6,0);
	\draw[dashed] (-1,0) sin (0,1) cos (1,0);
	\draw[fill,blue] (0,1)circle [radius=0.05cm];
	\draw (-6,-1.8) node[above]{$t=1$} -- (6,-1.8);
	\draw[dashed] (-2,-1.8) sin (-1,-1) cos (0,-1.8) sin (1,-2.6) cos (2,-1.8);
	\draw[fill,blue] (-1,-1)circle [radius=0.05cm];
	\draw[fill,blue] (1,-2.6)circle [radius=0.05cm];
	\draw (-6,-3.5) node[above]{$t=2$} -- (6,-3.5);
	\draw[dashed] (-3,-3.5) sin (-2,-2.9) cos (-1,-3.5) sin (0,-4.1) cos (1,-3.5) sin (2,-2.9) cos (3,-3.5);
	\foreach \x in {-2,2}
	\draw[fill,blue] (\x,-2.9)circle [radius=0.05cm];
	\draw[fill,blue] (0,-4.1)circle [radius=0.05cm];
%
	\draw (-6,-5) node[above]{$t=3$} -- (6,-5);
	\draw[dashed] (-4,-5) sin (-3,-4.6) cos (-2,-5) sin (-1,-5.4) cos (0,-5) sin (1,-4.6) cos (2,-5) sin (3,-5.4) cos (4,-5);
	\foreach \x in {-3,1}
	\draw[fill,blue] (\x,-4.6)circle [radius=0.05cm];
	\foreach \x in {-1,3}
	\draw[fill,blue] (\x,-5.4)circle [radius=0.05cm];
	\end{tikzpicture}
	\caption{Expansion of a particle (blue dot) along $\vec p$ (x-axis) with amplitudes (y-axis)}\label{step5B}
\end{figure}
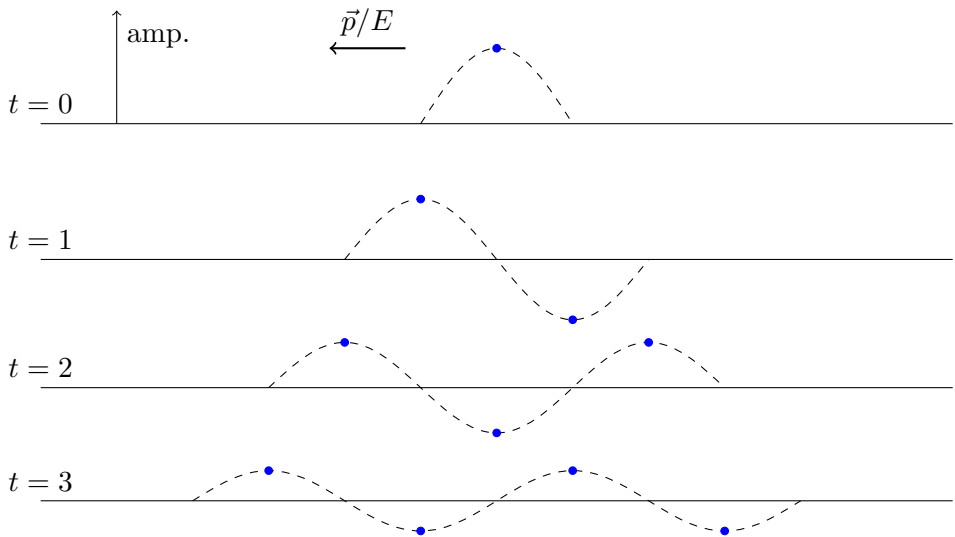

The amplitude acquires also a time dependent phase $e^{\iu Et}$ that makes
it complex.

\subsection{Fermions and Bosons}

The Lie algebras considered in this paper contain $\eo$, and thus $\dotto$.
Under the adjoint action of $\dotto$, the generators of $\eo$ split into
spinorial and non-spinorial ones, providing the algebra with a 2-graded
structure. We give the spinorial generators the physical meaning of
fermions, whereas the non-spinorial generators will be given the physical
meaning of bosons, in order to automatically comply with the addition of
angular momenta.

On the other hand, the Pauli exclusion principle is embodied in the \textit{%
Grassmann envelope} that turns the 2-graded algebra into a Lie superalgebra.
The degrees of freedom of the spin-1/2 fermions originate from the
superposition of opposite 3-momenta and the corresponding change of helicity
caused by the reflection at the space boundary. The Poincar\'{e} group is
then naturally emerging as a group of transformations of the \textit{local}
algebra $\ggu$ leaving the charges invariant.

All these topics will be treated in the companion paper \cite{mym2}

\subsection{Quantum Quasicrystal}

The expansion of the space that we propose has two fundamental features:

\begin{enumerate}
\item a space point may exist with a certain \textit{probability amplitude},
this latter being the sum of the amplitudes for some particles - matter or
radiation - to be there: no space point can possibly be empty;

\item space is a quantum object, that expands according to algebraic rules.
\end{enumerate}

Because of these two features our model of the Universe can be conceived as
a \textit{quantum quasicrystal}, \cite{QC1}\cite{QC2}\cite{QC3}.

\section{$\eo$, the Charge/Spin Subalgebra}

In our treatment, we use the following labels for the Dynkin diagram of $\eo$%
: \be
\setlength{\unitlength}{3pt}
\begin{picture}(85,15)
\linethickness{0.3mm}
\multiput(2,5)(8,0){7}{\circle{2}}
\put(34,13){\circle{2}}
\put(33.95,6){\line(0,1){6}}
\multiput(3,5)(8,0){6}{\line(1,0){6}}
\put(1,0){$\as_1$}
\put(9,0){$\as_2$}
\put(17,0){$\as_3$}
\put(25,0){$\as_4$}
\put(33,0){$\as_5$}
\put(41,0){$\as_7$}
\put(49,0){$\as_8$}
\put(36,12){$\as_6$}
\end{picture}
\label{dynkeo}
\ee

A way of writing the simple roots of $\eo$ in the orthonormal basis $%
\{k_{1},...,k_{8}\}$ of $\rr^{8}$ is:%
\begin{eqnarray}
\alpha _{1} &=&k_{1}-k_{2}\newline
;  \notag \\
\alpha _{2} &=&k_{2}-k_{3}\newline
;  \notag \\
\alpha _{3} &=&k_{3}-k_{4}\newline
;  \notag \\
\alpha _{4} &=&k_{4}-k_{5}\newline
;  \label{rsr8} \\
\alpha _{5} &=&k_{5}-k_{6}\newline
;  \notag \\
\alpha _{6} &=&k_{6}-k_{7}\newline
;  \notag \\
\alpha _{7} &=&k_{6}+k_{7}\newline
;  \notag \\
\alpha _{8} &=&-\frac{1}{2}\left(
k_{1}+k_{2}+k_{3}+k_{4}+k_{5}+k_{6}+k_{7}+k_{8}\right) .  \notag
\end{eqnarray}%
The whole root system $\fieo$ of $\eo$ (obtained from the simple roots by
Weyl reflections) can be written as follows:
\begin{equation}
\begin{array}{ll}
\Phi _{8}=\Phi _{\mathcal{O}}\cup \Phi _{\mathcal{S}}~\text{(}24\text{%
0~roots)}; &  \vspace{2pt}\\
\Phi _{\mathcal{O}}:=\left\{ \pm k_{i}\pm k_{j}|1\leqslant i<j\leqslant
8\right\} , & 4\binom{8}{2}=112; \vspace{2pt}\\
\Phi _{\mathcal{S}}:=\left\{ \frac{1}{2}\left( \pm k_{1}\pm k_{2}\pm ...\pm
k_{7}\pm k_{8}\right) ,~\text{even~\#~of~}+\right\} , & 2^{7}=128.%
\end{array}%
\end{equation}%
The first set $\fio$ of 112 roots is the set of roots of $\dotto\simeq
\mathbf{so(16,\cc)}$. The set $\fis$ is a Weyl spinor of $\dotto$ with
respect to the adjoint action (every orthogonal Lie algebra in even
dimension $\mathbf{d_{n}}\simeq \mathbf{so(2n,\cc)}$ has a Weyl spinor
representation of dimension $2^{n-1}$).

If $\as$ is a root, there is a unique way of writing it as $\as=\sum \lam_{i}%
\as_{i}$ where the $\as_{i}$'s are simple (in fact, all $\lam_{i}$'s are
positive for positive roots, and negative for negative roots). The sum $\hei(%
\as):=\sum \lam_{i}$ is called the {\bfit{height}} of $\as$.

The fact that the roots of $\fio$ are the roots of a subalgebra and those of
$\fis$ correspond to a representation of it can be seen by noticing that: $%
\fio+\fio\subset \fio$, $\fio+\fis\subset \fis$. Moreover, $\fis+\fis\subset %
\fio$ implies that $\fio$ is embedded into $\eo$ in a symmetric way.

Thus, one can consistently define a non-Cartan generators $x_{\as}$ bosonic
if $\as\in \fio$ and fermionic if $\as\in \fis$. We also call \textit{%
fermionic} or \textit{bosonic} the root $\as$ associated to a fermionic
(resp. bosonic) non-Cartan generator $x_{\as}$. A Cartan generator $h_{\as}$
is always bosonic for any $\as$ since $h_{\as}=[x_{-\as},x_{\as}]$. The
roots of $\eo$ split into 128 fermions (F) and 112 bosons (B).

\subsection{Algebraic Structure \label{seo}}

The $\eo$ algebra can be defined from its root system $\Phi _{8}$, \cite%
{carter, hum, graaf}, over the complex field extension $\mathbb{C}$ of the
rational integers $\mathbb{Z}$ in the following way:

\benal{}
\item we select the set of simple roots $\Delta _{8}$ of $\Phi _{8}$;

\item we select a basis $\left\{ h_{1},...,h_{8}\right\} $ of the
8-dimensional vector space $\mathfrak{h}$ over $\mathbb{C}$ and set $%
h_{\alpha }=\sum_{i=1}^{8}\lambda _{i}h_{i}$ for each $\alpha \in \Phi _{8}$
such that $\alpha =\sum_{i=1}^{8}\lambda _{i}\alpha _{i}$;

\item we associate to each $\alpha \in \Phi _{8}$ a one-dimensional
vector space $L_{\alpha }$ over $\mathbb{C}$ spanned by $x_{\alpha }$;

\item we define $\eo=\mathfrak{h}\bigoplus\limits_{\alpha \in \Phi
_{8}}L_{\alpha }$ as a vector space over $\mathbb{C}$;

\item we give $\eo$ an algebraic structure by defining the following
multiplication on the basis $\{h_{1},...,h_{8}\}\cup \{x_{\as}\ |\ \as\in
\Phi _{8}\}$, by linearity to a bilinear multiplication $\eo\times \eo%
\rightarrow \eo$:
\be\begin{array}{ll}
	&[h_i,h_j] = 0 \ , \ 1\le i, j  \le 8 \\
	&[h_i , x_\as] = - [x_\as , h_i] = (\as, \as_i )\, x_\as \ , \ 1\le i \le 8 \ , \ \as \in \Phi_8 \\
	&[x_\as, x_{-\as} ] = - h_\as\\
	&[x_\as,x_\bb] = 0 \ \text{for } \as, \bb \in \Phi_8 \ \text{such that } \as + \bb \notin			 	\Phi_8 \ \text{and } \as \ne - \bb\\
	&[x_\as,x_\bb] = \ep (\as , \bb)\,  x_{\as+\bb}\ \text{for } \as , \bb \in \Phi_8 \ \text{such that }  \as+ \bb \in \Phi_8\\
\end{array} \label{comrel}\ee

where $\ep(\as,\bb)$ is the \textit{asymmetry function}%
, introduced in \cite{kac} as in Definition \sref{asym}, see also \cite{graaf}.
\een

\begin{definition}\label{asym} Let $\Lleo$ denote the lattice of all linear combinations of the simple roots with integer coefficients
	\be
	\Lleo = \left\{ \sum_{i=1}^8 c_i \as_i \ |\ c_i \in \zz \ , \ \as_i \in \Delta_8 \right\}
	\label{lattice}
	\ee
	the asymmetry function $\ep (\as , \bb) : \ \Lleo \times \Lleo \to \{-1,1\}$ is defined by:
	\be\label{epsdef}
	\ep (\as , \bb) = \prod_{i,j=1}^8 \ep (\as_i , \as_j)^{\ell_i m_j} \quad \text{for } \as = \sum_{i=1}^8 \ell_i\as_i \ ,\ \bb = \sum_{j=1}^8 m_j \as_j
	\ee
	where $\as_i , \as_j \in \Delta_8$ and
	\be
	\ep (\as_i , \as_j) = \left\{
	\begin{array}{ll}
		-1 & \text{if } i=j\\ \\
		-1 & \text{if } \as_i + \as_j  \text{ is a root and } \as_i < \as_j\\ \\
		+ 1 & \text{otherwise}
	\end{array}
	\right.
	\ee
\end{definition}

We recall the following standard result on the roots of $\eo$ (normalized to
2), \cite{carter},\cite{hum}:

\begin{proposition} \label{sproots}
	For each $\as,\bb \in \Phi_8$ the scalar product $(\as,\bb) \in \{\pm 2, \pm 1, 0\}$; $\as + \bb$ ( respectively $\as - \bb$) is a root if and only if $(\as , \bb) = -1$ (respectively $+1$); if both $\as+\bb$ and  $\as-\bb$ are not in $\Phi_8\cup\{0\}$ then  $(\as,\bb)=0$.\\
	For $\as , \bb \in \Phi_8$ if $\as+\bb$ is a root then $\as - \bb$ is not a root.
\end{proposition}

The following properties of the asymmetry function follow from its
definition, \cite{graaf}.

\begin{proposition}\label{epsprop} The asymmetry function $\ep$ satisfies, for $\as , \bb, \gh \in \Lleo$:
	\bes\ba{rrcl}
	i) & \ep (\as + \bb, \gh) & =& \ep (\as , \gh)\ep (\bb , \gh) \\
	ii) &\ep (\as ,\bb+\gh) & = &\ep (\as , \bb)\ep (\as , \gh) \\
	iii) &\ep (\as , \as) & = & (-1)^{\frac12 (\as,\as)} \impl \ep (\as , \as) = -1 \text{ if } \as\in\Phi_8\\
	iv) &\ep (\as , \bb) \ep (\bb , \as) &=& (-1)^{(\as , \bb)} \impl \ep (\as , \bb) = -\ep (\bb , \as)\text{ if } \as,\bb,\as+\bb\in\Phi_8\\
	v) &\ep (0 , \bb) &=& \ep (\as , 0) = 1 \\
	vi) &\ep (-\as , \bb) &=& \ep (\as , \bb)^{-1} = \ep (\as , \bb) \\
	vii) &\ep (\as , -\bb) &=& \ep (\as , \bb)^{-1} = \ep (\as , \bb) \\
	\ea\ees
\end{proposition}

Property $iv)$ shows that the product in \eqref{comrel} is indeed
antisymmetric.

\subsection{$\eo$ Charges and the Magic Star\label{eoch}}

There are four orthogonal $\ad$'s in $\eo$, where orthogonal means that the
planes on which their root systems lie are orthogonal to each other.\newline
We denote one of them $\adc$ for color, one $\adf$ for flavor, the other two
$\adgu$ and $\adgd$:\\

\vspace{.3cm}
\noindent $\adc$ : $k_i - k_j$ , $i\ne j\ , \ i,j\in \{1,2,3\}$

\vspace{.3cm}
\noindent $\adf$ : $k_i - k_j$ , $i\ne j\ , \  i,j\in \{4,5,6\}$

\vspace{.3cm}
\noindent $\adgu$ : $\pm (k_7 + k_8)$, $\pm \frac{1}{2} ( k_1 + k_2 + k_3 + k_4  + k_5  + k_6 - k_7 - k_8)$ \\
\phantom{$\adgu$ : }$\pm \frac{1}{2} ( k_1 + k_2 + k_3 + k_4  + k_5  + k_6 + k_7 + k_8)$

\vspace{.3cm}
\noindent $\adgd$ : $\pm (k_7 - k_8)$ , $\pm \frac{1}{2} ( - k_1 + k_2 - k_3 + k_4 - k_5 + k_6 - k_7 + k_8)$\\
\phantom{$\adgd$ : }$\pm \frac{1}{2} ( - k_1 + k_2 - k_3  + k_4 - k_5 + k_6 + k_7 - k_8)$\\

The generators of $\adc$ and $\adf$ are bosonic.\newline

\begin{figure}[hbtp]
	\begin{center}
		\begin{tikzpicture}
		\coordinate (O) at (0,0);
		\coordinate (a1) at (2,0);
		\coordinate (a1m) at (-2,0);
		\coordinate (a2) at (-1,1.8);
		\coordinate (a2m) at (1,-1.8);
		\coordinate (a3) at (1,1.8);
		\coordinate (a3m) at (-1,-1.8);
		
		\coordinate (j11) at (1,.6);
		\coordinate (j1m1) at (-1,.6);
		\coordinate (j11m) at (1,-.6);
		\coordinate (j1m1m) at (-1,-.6);
		\coordinate (j02) at (0,1.2);
		\coordinate (j02m) at (0,-1.2);
		
		\draw[magenta] (O) circle [radius=1mm];
		\fill[magenta] (O) circle [radius=.5mm];
		\fill[magenta] (a1) circle [radius=.5mm];
		\fill[magenta] (a1m) circle [radius=.5mm];
		\fill[magenta] (a2) circle [radius=.5mm];
		\fill[magenta] (a2m) circle [radius=.5mm];
		\fill[magenta] (a3) circle [radius=.5mm];
		\fill[magenta] (a3m) circle [radius=.5mm];
		\fill[blue] (j11) circle [radius=.7mm];
		\fill[blue] (j11m) circle [radius=.7mm];
		\fill[blue] (j1m1) circle [radius=.7mm];
		\fill[blue] (j1m1m) circle [radius=.7mm];
		\fill[blue] (j02) circle [radius=.7mm];
		\fill[blue] (j02m) circle [radius=.7mm];
		\node at (-.2,.22){$\mathbf {g_o}$};
		\node at (2.5,.0){$\scriptstyle \{2,0\}$};
		\node at (-2.5,0){$\scriptstyle \{\text{-}2,0\}$};
		\node at (1,2.1){$\scriptstyle \{1,3\}$};
		\node at (-1,2.1){$\scriptstyle \{\text{-}1,3\}$};
		\node at (-1,-2.1){$\scriptstyle \{\text{-}1,\text{-}3\}$};
		\node at (1,-2.1){$\scriptstyle \{1,\text{-}3\}$};
		
		\node at (1,.9){$\scriptstyle \{1,1\}$};
		\node at (-1,.9){$\scriptstyle \{\text{-}1,1\}$};
		\node at (-1,-.9){$\scriptstyle \{\text{-}1,\text{-}1\}$};
		\node at (1,-.9){$\scriptstyle \{1,\text{-}1\}$};
		\node at (0,1.5){$\scriptstyle \{0,2\}$};
		\node at (0,-1.5){$\scriptstyle \{0,\text{-}2\}$};
		\draw[dotted] (j1m1) -- (j11);
		\draw[dotted] (j11) -- (j02m);
		\draw[dotted] (j02m) -- (j1m1);
		\draw[dotted] (j1m1m) -- (j11m);
		\draw[dotted] (j11m) -- (j02);
		\draw[dotted] (j02) -- (j1m1m);
		\draw[dotted] (a1) -- (a3);
		\draw[dotted] (a3) -- (a2);
		\draw[dotted] (a2) -- (a1m);
		\draw[dotted] (a1m) -- (a3m);
		\draw[dotted] (a3m) -- (a2m);
		\draw[dotted] (a2m) -- (a1);
		\end{tikzpicture}
	\end{center}
	\caption{The Magic Star (MS): $\ggo=\es$ in the MS of $\eo$, $\ggo=\adgu\oplus\adgd$ in the MS of $\es$; the triangles represent the $\mathbf 3$ and $\mathbf{\overline 3}$ representations of the $\ad$ with roots in the external hexagon.}\label{f:ms}
\end{figure}
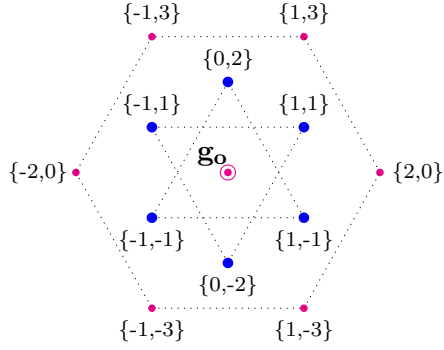

The Magic Star (MS) of $\eo$ shown in Fig. \ref{f:ms} is obtained by
projecting its roots on the plane of $\adc$, \cite{pt1}. The pair of
integers $\{r,s\}$ are the (Euclidean) scalar products $r:=(\as,k_{1}-k_{2})$
and $s:=(\as,k_{1}+k_{2}-2k_{3})$, for each root $\as$. The fermions on the
tips of the MS are quarks since they are acted upon by $\adc$: they are
\textit{colored}. The fermions within the center of the MS are leptons: they
are colorless. A similar MS of $\es$ within $\eo$ is obtained by projecting
the roots in the center of the MS of $\eo$ on the plane of $\adf$.\newline

Notice that in each tip of the MS of $\eo$ we get 27 roots, 11 of which are
bosonic and 16 fermionic; this corresponds to the following decomposition of
the irrepr. $\mathbf{27}$ of $\es$ :%
\begin{equation}
\begin{array}{l}
\es\supset \mathbf{d}_{\mathbf{5}}\oplus \mathbb{C}, \\
\mathbf{27}=\mathbf{1}_{-4}\oplus \mathbf{10}_{2}\oplus \mathbf{16}_{-1}.%
\end{array}%
\end{equation}%
On the other hand, within $\es$ we have 9 roots in each tip of the MS, 5 of
which are bosonic and 4 fermionic; this corresponds to the following
decomposition of the repr. $\left( \mathbf{3,3}\right) $ of $\ad\oplus \ad$ :%
\begin{equation}
\begin{array}{l}
\ad\oplus \ad\supset \mathfrak{a}_{1}\oplus \mathfrak{a}_{1}\oplus \mathbb{C}%
_{I}\oplus \mathbb{C}_{II}\supset \mathfrak{a}_{1}\oplus \mathfrak{a}%
_{1}\oplus \mathbb{C}_{I+II}, \\
\left( \mathbf{3,3}\right) =\left( \mathbf{1,1}\right) _{-2,-2}\oplus \left(
\mathbf{2,2}\right) _{1,1}\oplus \left( \mathbf{2,1}\right) _{1,-2}\oplus
\left( \mathbf{1},\mathbf{2}\right) _{-2,1}\\
\php{\left( \mathbf{3,3}\right)} = \left( \mathbf{1,1}\right)
_{-4}\oplus \left( \mathbf{2,2}\right) _{2}\oplus \left( \mathbf{2,1}\right)
_{-1}\oplus \left( \mathbf{1},\mathbf{2}\right) _{-1}.%
\end{array}%
\end{equation}

Tables \ref{t:magiceo}, \ref{t:magices}, \ref{t:leptons}, \ref{t:quarks}
describe the content of these MS, root by root.\newline

The \textit{magic} of the MS is that each tip $\{r,s\}$ of the star, $%
\{r,s\}\in \{\{0,\pm 2\},\{\pm 1,\pm 1\}\}$, both in the case of $\eo$ and
of $\es$, can be viewed as a cubic (simple) \textit{Jordan algebra} $%
J_{\{r,s\}}$, over the octonions and the complex field respectively, and
each pair of opposite tips with respect to the center of the star has a
natural algebraic structure of a \textit{Jordan pair}. The algebra in the
center of the star is the derivation algebra of the Jordan pair; when the
Jordan pair is made of a pair of Jordan algebras, its derivations also
define the Lie algebra of the structure group of the Jordan algebra itself,
\cite{loos, bied1, pt1, T-2}.

\subsection{The Standard Model\label{stdm}}

In this section, we relate the $\eo$ charges to the degrees of freedom of
the Standard Model (SM) of elementary particle physics. It is not our aim to
carry through a detailed analysis, in particular we do not consider symmetry
breaking, nor Higgs mechanism, nor chirality and parity violation by weak
interactions in the fermionic sector. We do however focus on spin as an
internal degree of freedom, and this will be instrumental for the treatment
of the Poincar\'{e} action on our algebra, which we will investigate in the
companion paper, \cite{mym2}.\newline

The first important step, after splitting the roots into colored and
colorless as in the previous section, is to find the electromagnetic $%
\mathbf{u(1)}^{em}$ that gives the right charges to quarks and leptons. We
select the one generated by%
\begin{equation}
h_{\gh}=-\iu(\frac{1}{3}h_{\as_{1}}+\frac{2}{3}h_{\as_{2}}+h_{\as_{3}}),
\end{equation}
giving to $x_{\as}$, where $\as=\sum \lam_{i}k_{i}$, the charge%
\begin{equation}
q_{e.m.}(\as):=(\as,q_{\gamma })=-\frac{1}{3}(\lam_{1}+\lam_{2}+\lam_{3})+%
\lam_{4},
\end{equation}%
where%
\begin{equation}
q_{\gamma }:=-\frac{1}{3}\as_{1}-\frac{2}{3}\as_{2}-\as_{3}=-\frac{1}{3}%
(k_{1}+k_{2}+k_{3})+k_{4}.
\end{equation}%
The second column of Tables \ref{t:magiceo}, \ref{t:magices}, \ref{t:leptons}%
, \ref{t:quarks} shows the charges of the $\eo$ generators $x_{\as}$, with $%
\as$ shown in the first column. In particular, Tables \ref{t:leptons} and %
\ref{t:quarks} show the charges given to quarks and leptons.

We now select the semisimple Lie algebra $\mathbf{d_{2}}\simeq \mathbf{so(4,%
\cc)}\simeq \au\oplus \au$ with roots $\pm k_{5}\pm k_{6}$ (other choices
would be equivalent, thus this does not imply any loss of generality). We
denote by $\rho _{1},\rho _{2}$ the roots $k_{5}-k_{6}$ and $k_{5}+k_{6}$,
respectively, and by $\au^{(1)}$ the subalgebra of $\mathbf{d_{2}}$
associated to the roots $\pm \rho _{1}$ and by $\au^{(2)}$ the one
associated to the roots $\pm \rho _{2}$. The non-Cartan generators of $\eo$
fall into $\au\oplus \au$ irreducible representations of spin $(s_{1},s_{2})$%
:%
\begin{equation}
\begin{array}{ll}
\left( 0,0\right) : & x_{\alpha },\alpha =\pm k_{i}\pm k_{j},\ i,j\notin
\left\{ 5,6\right\} ; \\
\left( 1/2,0\right) : & x_{\alpha },\alpha =\frac{1}{2}\left( k\pm \left(
k_{5}-k_{6}\right) \right) , \ k:=\pm k_{1}\pm k_{2}\pm k_{3}\pm k_{4}\pm
k_{7}\pm k_{8}; \\
\left( 0,1/2\right) : & x_{\alpha },\alpha =\frac{1}{2}\left( k\pm \left(
k_{5}+k_{6}\right) \right) ; \\
\left( 1/2,1/2\right) : & x_{\alpha },\alpha =\pm k_{i}\pm k_{5}~\text{or~}%
\alpha =\pm k_{i}\pm k_{6}, \ \text{for~a~fixed~}i\notin \left\{ 5,6\right\}
; \\
\left( 1,0\right) : & x_{\alpha },\alpha =\pm \left( k_{5}-k_{6}\right) , \
\text{the~adjoint~of~$\au^{(1)}$ (adding its Cartan);} \\
\left( 0,1\right) : & x_{\alpha },\alpha =\pm \left( k_{5}+k_{6}\right) , \
\text{the adjoint of $\au^{(2)}$ (adding its Cartan)}\newline
,%
\end{array}%
\end{equation}%
where $(1,0)\oplus (0,1)$ corresponds to the 6 components as a rank 2
antisymmetric tensor in 4 dimensions, with \textit{selfdual} and
\textit{antiselfdual} parts $(1,0)$ and $(0,1)$, respectively. Notice that all
fermions have half-integer spin $(\frac{1}{2},0)$ or $(0,\frac{1}{2})$,
whereas all bosons of type $x_{\as}$ have integer spin.

In order to define the action of the Poincar\'{e} group in \cite{mym2}, we
need the covering group of rotations in the internal space. For this purpose
we select the spin (diagonal) subalgebra $\spin\in \au^{(1)}\oplus \au^{(2)}$ as the
compact (real) form with generators
\bea{l}\label{rpm}
R^{+}+R^{-},\quad \iu \left(R^{+}-R^{-}\right), \quad \iu H_{R}, \quad \text{where}\\ R^{+}:=x_{\rho _{1}}+x_{\rho _{2}},\quad R^{-}:=x_{-\rho _{1}}+x_{-\rho _{2}},\quad H_{R}:=\frac{1}{2}(h_{\rho
_{1}}+h_{\rho _{2}})
\eea
The $(\frac{1}{2},\frac{1}{2})$ representation splits
into a scalar and a vector under this rotation subalgebra. The spin-1
particle within $(\frac{1}{2},\frac{1}{2})$ is the linear span of the
generators $x_{k_{i}\pm k_{5}}$ with z-component of spin $s_{z}:=\frac{1}{2}%
(\rho _{1}+\rho _{2},k_{i}\pm k_{5})=(k_{5},k_{i}\pm k_{5})=\pm 1$ and $%
\frac{1}{2}(\eps(\rho _{1},k_{i}-k_{5})x_{k_{i}-k_{6}}+\eps(\rho
_{2},k_{i}-k_{5})x_{k_{i}+k_{6}})$ with $s_{z}=0$; the corresponding scalar
is $\frac{1}{2}(\eps(\rho _{1},k_{i}-k_{5})x_{k_{i}-k_{6}}-\eps(\rho
_{2},k_{i}-k_{5})x_{k_{i}+k_{6}})$, as it can be easily verified.

\subsubsection{$W^\pm$}

Let us now consider the $W^{\pm }$ bosons. There are not many choices for
them: indeed, they must be colorless vectors with respect to $\spin$ and
have electric charge $\pm 1$. The $W^{\pm }$ bosons are therefore the
generators associated to $\pm (k_{4}-k_{5})$ (within $\adf$ mentioned in
section \sref{eoch}), $\pm (k_{4}+k_{5})$, the electric charge $\pm 1$ given
by the presence of $k_{4}$; they change flavor to both quarks and leptons.
The above analysis suggests that the extra degree of freedom needed, say for
$W^{+}$, to become massive, from the 2 degree of freedom of the massless
helicity-1 state, is $\eps(\rho _{1},k_{4}+k_{6})x_{k_{4}+k_{6}}+\eps(\rho
_{2},k_{4}-k_{6})x_{k_{4}-k_{6}}$, as a part of the Higgs mechanism - that
we will not discuss any further in this paper.

\begin{remark}\label{remdg} We could have made other equivalent choices for the $\mathbf{d_2}\simeq \mathbf{so(4,\cc)} \simeq \au\oplus\au$ subalgebra, which has to act non-trivially on $W^\pm$: once passing to real forms, it cannot possibly commute with the weak interaction $\sud^L$. For what concerns the no-go theorem by Distler Garibaldi, the hypothesis TOE1 of \cite{dg} cannot possibly apply, as outlined in the introduction, see remark \sref{rdg}. We also emphasize that, contrary to \cite{dg}, we are dealing with the complex form of $\eo$ because we want complex phases for the particle states.
\end{remark}

Using the properties of the asymmetry function and the ordering of the
simple roots $\as_{i}<\as_{i+1}$ we get:%
\begin{equation}
\eps(\rho _{1},k_{4}-k_{5})=\eps(\rho _{2},k_{4}-k_{5})=1,
\end{equation}
hence the massive $W^{\pm }$ is described by three components:%
\begin{equation}
\begin{array}{ll}
W_{1}^{\pm }:=x_{\pm k_{4}+k_{5}} & (s_{z}=1); \\
W_{0}^{\pm }:=\frac{1}{2}(x_{\pm k_{4}-k_{6}}+x_{\pm k_{4}+k_{6}}) & \
(s_{z}=0); \\
W_{-1}^{\pm }:=x_{\pm k_{4}-k_{5}} & (s_{z}=-1).%
\end{array}%
\end{equation}

Moreover, using the notation%
\begin{equation}
\begin{array}{l}
R_{x}:=\frac{1}{2}(x_{\rho _{1}}+x_{\rho _{2}}+x_{-\rho _{1}}+x_{-\rho
_{2}}); \\
R_{y}:=\frac{\iu}{2}(x_{\rho _{1}}+x_{\rho _{2}}-x_{-\rho _{1}}-x_{-\rho
_{2}}); \\
R_{z}:=\frac{\iu}{2}(h_{\rho _{1}}+h_{\rho _{2}}),%
\end{array}%
\end{equation}
we obtain%
\begin{equation}
\begin{array}{lllll}
\left[ R_{x},W_{-1}^{\pm }\right] =W_{0}^{\pm }, &  & \left[
R_{x},W_{1}^{\pm }\right] =W_{0}^{\pm }, &  & \left[ R_{x},W_{0}^{\pm }%
\right] =-\frac{1}{2}\left( W_{1}^{\pm }+W_{-1}^{\pm }\right) , \vspace{4pt}\\
\left[ R_{y},W_{-1}^{\pm }\right] =\iu W_{0}^{\pm }, &  & \left[
R_{y},W_{1}^{\pm }\right] =-\iu W_{0}^{\pm }, &  & \left[ R_{y},W_{0}^{\pm }%
\right] =-\frac{\iu}{2}\left( W_{1}^{\pm }-W_{-1}^{\pm }\right) ,%
\end{array}
\label{rw}
\end{equation}%
and%
\begin{equation}
\left[ R_{z},W_{s_{z}}^{\pm }\right] =\iu s_{z}W_{s_{z}}^{\pm },~s_{z}\in
\left\{ 1,0,-1\right\} .
\end{equation}%
These commutation relations correspond to the action of the rotation
matrices:%
\begin{equation}
\begin{array}{ccccc}
J_{x}:=\frac{1}{\sqrt{2}}\left(
\begin{array}{ccc}
0 & -1 & 0 \\
1 & 0 & 1 \\
0 & -1 & 0%
\end{array}%
\right) , &  & J_{y}:=\frac{\iu}{\sqrt{2}}\left(
\begin{array}{ccc}
0 & -1 & 0 \\
-1 & 0 & 1 \\
0 & 1 & 0%
\end{array}%
\right) , &  & J_{z}:=\iu\left(
\begin{array}{ccc}
1 & 0 & 0 \\
0 & 0 & 0 \\
0 & 0 & -1%
\end{array}%
\right)%
\end{array}%
\end{equation}%
on the vectors $v=v_{+}e_{+}+v_{z}e_{z}+v_{-}e_{-}\in \mathbb{R}^{3}$ in the
spherical basis $\left\{ e_{+},e_{z},e_{-}\right\} $ corresponding, for
angular momentum 1, to the spherical harmonic basis for the irreducible
representations of $SO(3)$. With respect to the same vector $%
v=v_{x}e_{x}+v_{y}e_{y}+v_{z}e_{z}\in \mathbb{R}^{3}$ in the standard
orthogonal basis $\left\{ e_{x},e_{y},e_{z}\right\} $, we have:%
\begin{equation}
\begin{array}{lll}
e_{+}=\frac{1}{\sqrt{2}}\left( e_{x}+\iu e_{y}\right) , &  & e_{-}=\frac{1}{%
\sqrt{2}}\left( e_{x}-\iu e_{y}\right) , \\
v_{+}=\frac{1}{\sqrt{2}}\left( v_{x}-\iu v_{y}\right) , &  & v_{-}=\frac{1}{%
\sqrt{2}}\left( v_{x}+\iu v_{y}\right) .%
\end{array}%
\end{equation}%
The transformation between (column) vectors in the two bases is represented
by the unitary matrix $U$:%
\begin{equation}
\left(
\begin{array}{c}
v_{+} \\
v_{z} \\
v_{-}%
\end{array}%
\right) =U\left(
\begin{array}{c}
v_{x} \\
v_{y} \\
v_{z}%
\end{array}%
\right) ,~~~U:=\left(
\begin{array}{ccc}
\frac{1}{\sqrt{2}} & 0 & -\frac{\iu}{\sqrt{2}} \\
0 & 1 & 0 \\
\frac{1}{\sqrt{2}} & 0 & \frac{\iu}{\sqrt{2}}%
\end{array}%
\right) .
\end{equation}%
The correspondence with $R$'s and $W$'s is:%
\begin{equation}
\begin{array}{ccccc}
R_{x}\leftrightarrow J_{x}, &  & R_{y}\leftrightarrow J_{y}, &  &
R_{z}\leftrightarrow J_{z}, \vspace{4pt}\\
W_{1}^{\pm }\leftrightarrow v_{+}, &  & W_{-1}^{\pm }\leftrightarrow v_{-},
&  & W_{0}^{\pm }\leftrightarrow \frac{1}{\sqrt{2}}v_{z}.%
\end{array}%
\end{equation}

\begin{remark}
We have an interesting relationship between weak and rotation generators in the internal space (spin generators) by noticing that
\bea{cl}
\left[W^+_1+W^+_{-1},\php{\iu} W^-_0\right]=\left[W^-_1+W^-_{-1},\php{\iu}W^+_0\right]&=R_x\vspace{4pt} \\
\left[W^+_1-W^+_{-1},\iu W^-_0\right]=\left[W^-_1-W^-_{-1},\iu W^+_0\right]&=R_y
\eea
and consequently, the relation with $R_z=[R_x,R_y]$.
\end{remark}

\subsubsection{$Z^0$}

We associate the $Z^{0}$ boson with spin $s_{z}=0$, denoted by $Z_{0}^{0}$,
to the vector orthogonal to $q_{\gamma }$ in the plane of $(k_{4}-k_{5})$
and $q_{\gamma }$, hence it is, up to a scalar, the Cartan generator $%
Z_{0}^{0}=\frac{1}{4}(h_{\as_{1}}+2h_{\as_{2}}+3h_{\as_{3}}+4h_{\as_{4}})$;
it interacts with left-handed neutrinos and right-handed antineutrinos,
contrary to the photon; it does not allow for \textit{flavor changing
neutral currents}.

Notice that the generator of \textit{hypercharge} $\mathbf{u(1)}^{Y}$ of the
standard model is in this setting compact Cartan generator $\iu h_{Y}$ where
$h_{Y}:=-\frac{1}{6}(2h_{\as_{1}}+4h_{\as_{2}}+6h_{\as_{3}}+3h_{\as_{4}})$.
The Weinberg angle $\theta _{W}$ is the angle between the axis representing
the photon and that representing the hypercharge, therefore $\theta _{W}=\pi
/2-\phi $ where $\phi $ is the angle between $q_{\gamma }$ and $k_{4}-k_{5}$%
; we get $sin^{2}\theta _{W}=3/8$.

Since $\as_{1}+2\as_{2}+3\as_{3}+4\as_{4}=k_{1}+k_{2}+k_{3}+k_{4}-4k_{5}$
hence $\frac{1}{4}(\as_{1}+2\as_{2}+3\as_{3}+4\as_{4},x_{\pm \rho
_{1,2}})=-(k_{5},\pm \rho _{1})=-(k_{5},\pm \rho _{2})=\mp 1$, we have the
following commutation relations:%
\begin{equation}
\begin{array}{ccc}
\lbrack x_{\rho _{1}}+x_{\rho _{2}}+x_{-\rho _{1}}+x_{-\rho _{2}},Z_{0}^{0}]
& = & (x_{\rho _{1}}+x_{\rho _{2}}-x_{-\rho _{1}}-x_{-\rho _{2}})\newline
; \vspace{2pt}\\
\left[ x_{\rho _{1}}+x_{\rho _{2}}-x_{-\rho _{1}}-x_{-\rho _{2}},Z_{0}^{0}%
\right] & = & (x_{\rho _{1}}+x_{\rho _{2}}+x_{-\rho _{1}}+x_{-\rho _{2}}),%
\end{array}%
\end{equation}
that is%
\begin{equation}
\begin{array}{ccccc}
\left[ R_{x},Z_{0}^{0}\right] =-\iu R_{y}, &  & \left[ R_{y},Z_{0}^{0}\right]
=\iu R_{x}, &  & \left[ R_{z},Z_{0}^{0}\right] =0.%
\end{array}%
\end{equation}

We want $Z^{0}$, as a spin-1 particle, to obey the same commutation
relations with the rotation generators as $W^{\pm }$. We can so define the
spin $\pm 1$ components of $Z^{0}$ by comparison with the last commutator in
each row of \eqref{rw}:%
\begin{equation}
\begin{array}{lll}
Z_{1}^{0}= & -\left[ R_{x},Z_{0}^{0}\right] +\iu\left[ R_{y},Z_{0}^{0}\right]
= & -R_{x}+\iu R_{y}=-R_{+}; \vspace{4pt}\\
Z_{-1}^{0}= & -\left[ R_{x},Z_{0}^{0}\right] -\iu\left[ R_{y},Z_{0}^{0}%
\right] = & R_{x}+\iu R_{y}=R_{-}.%
\end{array}%
\end{equation}
(see \eqref{rpm} for the definition of $R_\pm$).

By looking at Table \ref{t:leptons} we notice that $Z^0_1$ interacts, for
instance, with $\nu^\prime_e$ with spin $-\frac12$ to give $\nu^\prime_e$
with spin $\frac12$, and similarly for other leptons and for quarks. In
particular there are no flavor changing neutral currents.

\subsubsection{The Tables at a Glance}

From the Tables \ref{t:magiceo}, \ref{t:magices}, \ref{t:leptons}, \ref%
{t:quarks} one can deduce all the standard model charges (in particular we
have denoted with a \textit{prime} possible mixings in Tables \ref{t:leptons}%
, \ref{t:quarks}). We have:

\benar{SM.}
\item the color charges are denoted by the pair $\{r_{c},s_{c}\}$ and
one can associate \textit{colors} to them, say $blue=\{1,1\}$, $%
green=\{-1,1\}$ and $red=\{0,-2\}$, and similarly for the anti-colors;

\item the quarks are the fermions in Table \ref{t:magiceo} with a
certain color; they come in 3 color families and anti-quarks have
anti-colors and opposite electric charges $-\frac{2}{3},\frac{1}{3}$ with
respect to quarks;

\item the gluons are the generators of $\adc$, change color to the
quarks on which they act as on a $\mathbf{3}$ or $\mathbf{\bar{3}}$
representation and their electric charge is $0$;

\item the leptons are in the center of the MS in Table \ref{t:magiceo}
and are the fermions in Table \ref{t:magices};

\item the leptons have integer electric charge in $\{-1,0,1\}$;

\item there are 4 flavor families; we have used the notation $\chi
,\nu _{\chi }$ for the fourth lepton family and $T,B$ for the fourth quark
family;

\item the fourth column of Tables \ref{t:leptons} and \ref{t:quarks}
shows the component of spin along the axis specified by the spin generator $%
R_{z}:=\frac{\iu}{2}(h_{\rho _{1}}+h_{\rho _{2}})$. Obviously a rotation by $%
2\pi $ of quarks and leptons changes their sign, whether it leaves vector
bosons invariant.
\een

The consequences of this classification with respect to the Poincar\'{e}
action on $\ggu$ will be discussed in the companion paper \cite{mym2}.

\section{The Kac-Moody Algebras $\eno,\ed,\edo,\dedo$\label{KM}}

Let $\ga$ denote the Kac-Moody algebra associated to the $n\times n$ Cartan
matrix $A$, with Cartan subalgebra $\hh$. For all algebras in this paper, $A$
is symmetric; its entries are denoted by $a_{ij}$. We denote the \textit{%
Chevalley generators} by $E_{i}$, associated to the simple root $\as_{i}$,
and by $F_{i}$, associated to the root $-\as_{i}$. Let $\np$ (resp. $\nm$)
denote the subalgebra of $\ga$ generated by $\{E_{1},...,E_{n}\}$ (resp.$%
\{F_{1},...,F_{n}\}$). By Theorem 1.2 a), e) in \cite{kac}, the following
\textit{triangular decomposition} holds :%
\begin{equation}
\ga=\np\oplus \hh\oplus \nm\qquad \text{(direct sum of vector spaces)}.
\label{tridec}
\end{equation}
Note that for a root $\as>0$ (resp. $\as<0$) we have $\as\in \hhs$, the dual
of $\hh$, and the vector space $\gas=\{x\in \ga\ |\ [h,x]=\as(h)x\ \forall
h\in \hh\}$ is the linear span of the elements of the form $%
[...[[E_{i_{1}},E_{i_{2}}],E_{i_{3}}]...E_{i_{s}}]$ (resp. $%
[...[[F_{i_{1}},F_{i_{2}}],F_{i_{3}}]...F_{i_{s}}]$), such that $\as%
_{i_{1}}+...+\as_{i_{s}}=\as$ (resp. $=-\as$). The multiplicity $m_{\as}$ of
a root $\as$ is defined as $m_{\as}:=\text{dim}\ \gas$ ($m_{\as_{i}}=m_{-\as%
_{i}}=1$ for each simple root $\as_{i}$).

Kac-Moody algebras, \cite{kac} \cite{moo}, can be tackled in terms of simple
roots and their (extended) Dynkin diagram, or equivalently their Cartan
matrix, without any reference to root coordinates. Some physical features or
interpretations may however be more explicit when roots are expressed in an
orthonormal basis rather than in a simple root basis. This is the case of
this paper, in which some root coordinates, except for the case of $\eo$,
are interpreted as momentum coordinates. We recall that the metric is
Euclidean for $\eo$ but Lorentzian in the case of $\eno,\ed,\edo,\dedo$.

Our notation for the simple roots is shown in the Dynkin diagram of $\edo$:
\be
\setlength{\unitlength}{3pt}
\begin{picture}(85,15)
\linethickness{0.3mm}
\multiput(2,5)(8,0){11}{\circle{2}}
\put(66,13){\circle{2}}
\put(65.95,6){\line(0,1){6}}
\multiput(3,5)(8,0){10}{\line(1,0){6}}
\put(0,0){$\rmu$}
\put(8,0){$\ropp$}
\put(17,0){$\rop$}
\put(25,0){$\as_0$}
\put(33,0){$\as_1$}
\put(41,0){$\as_2$}
\put(49,0){$\as_3$}
\put(57,0){$\as_4$}
\put(65,0){$\as_5$}
\put(73,0){$\as_7$}
\put(81,0){$\as_8$}
\put(68,12){$\as_6$}
\end{picture}
\label{dynkedo}
\ee
Analogous diagrams are those of $\ed$ (without $\rmu, \ropp$), of $\eno$
(without $\rmu, \ropp, \rop$) and of $\eo$ (without $\rmu, \ropp, \rop, \as%
_0 $).

\subsection{The Simple Roots of $\eno$ and $\ed$}

We introduce the following set of simple roots $\rmu,\as_{0},\as_{1},...,\as%
_{8}$ of $\ed$ in terms of the basis vectors $\km,k_{0},k_{1},...,k_{8}$
spanning the Lorentzian space $\rr^{9,1}$, with $(\km,\km)=-1$ and $%
(k_{i},k_{i})=1$, for $0\leq i\leq 8$:%
\bea{lcl}
\alpha _{-1} &=&\frac{1}{2}\left( k_{-1}-3k_{0}\right) ;   \\
\alpha _{0} &=&\frac{1}{2}\left( k_{-1}+k_{0}\right) -k_{1}+k_{8};   \\
\alpha _{1} &=&k_{1}-k_{2}
;   \\
\alpha _{2} &=&k_{2}-k_{3}
;   \\
\alpha _{3} &=&k_{3}-k_{4}
;  \label{deltaeno} \\
\alpha _{4} &=&k_{4}-k_{5}
;   \\
\alpha _{5} &=&k_{5}-k_{6}
;   \\
\alpha _{6} &=&k_{6}-k_{7}
;   \\
\alpha _{7} &=&k_{6}+k_{7}
;   \\
\alpha _{8} &=&-\frac{1}{2}\left(
k_{1}+k_{2}+k_{3}+k_{4}+k_{5}+k_{6}+k_{7}+k_{8}\right) .
\eea
All these roots have norm 2, with respect to the scalar product $(\cdot
,\cdot )$, and the corresponding Cartan matrix is the Gram matrix of the $\ed$
even unimodular Lorentzian lattice $\iinu$ made of all the vectors in $\rr%
^{9,1}$ whose components are all in $\zz$ or all in $\zz+\um$ and have
integer scalar product with $\um(\km%
+k_{0}+k_{1}+k_{2}+k_{3}+k_{4}+k_{5}+k_{6}+k_{7}+k_{8})$, as it can be
easily checked.

The \textit{affine} Kac-Moody algebra $\eno$ is obtained by eliminating the
root $\rmu$. Notice that%
\begin{equation}
\as_{0}=\as_{9}+\delta \text{ , }\left\{
\begin{array}{l}
\delta :=\frac{1}{2}(\km+k_{0}),\quad \text{(hence }(\delta ,\delta )=0\text{%
)}; \\
\as_{9}:=-(2\as_{1}+3\as_{2}+4\as_{3}+5\as_{4}+6\as_{5}+3\as_{6}+4\as_{7}+2%
\as_{8}),%
\end{array}%
\right.  \label{r9d}
\end{equation}
where $\as_{9}=-k_{1}+k_{8}$ is the lowest height root of $\eo$ and $\delta $
is a lightlike vector.

\subsection{$\eno$}

The Cartan subalgebra of $\eno$ is the span of 10 generators, two new ones
with respect to $\eo$. We write
\begin{equation}
\hh=\text{span}\{K,d,h_{i}\ |\ i=1,...,8\},
\end{equation}
where
\bea{lcl}
K &:=&h_{\delta },~~\delta :=\frac{1}{2}\left( k_{-1}+k_{0}\right) ;
\\
d &:=&h_{\rho },~~\rho :=-k_{-1}+k_{0},  \label{Kd}
\eea
and $K$ is a central element.

Let $\Phi _{8}$ be the root system of $\eo$, $\hh$ its Cartan subalgebra, $%
h_{\as}$ (resp. $e_{\as}$) a Cartan (resp. non-Cartan) generator associated
to the root $\as$, $h_{i}$ for $1\leq i\leq 8$ a Cartan generator associated
to the simple root $\as_{i}$, $X$, $Y$ either Cartan or non-Cartan
generators of $\eo$ . It is shown by Kac, \cite{kac}, that:

\begin{enumerate}
\item The root system $\Phi _{9}$ of $\eno$ is%
\begin{equation}
\Phi _{9}:=\{\as+m\delta \ |\ \as\in \Phi _{8}\ ,\ m\in \zz\}\cup \{m\delta
\ |\ m\in \zz\backslash 0\}.
\end{equation}

\item $\eo$ is determined by the following commutation relations:%
\begin{equation}
\begin{array}{llll}
\left[ h,h^{\prime }\right] & = & 0 & \text{if~}h,h^{\prime }\in \mathfrak{h}%
; \\
\left[ h,E_{\alpha }\right] & = & \left( h^{\ast },\alpha \right) E_{\alpha }
& \text{if~}h\in \mathfrak{h},~\alpha \in \Phi _{8}; \\
\left[ E_{\alpha },E_{-\alpha }\right] & = & -h_{\alpha } & \text{if~}\alpha
\in \Phi _{8}; \\
\left[ E_{\alpha },E_{\beta }\right] & = & 0 & \text{if~}\alpha ,\beta \in
\Phi _{8},~\alpha +\beta \notin \Phi _{8}\cup \left\{ 0\right\} ; \\
\left[ E_{\alpha },E_{\beta }\right] & = & \varepsilon \left( \alpha ,\beta
\right) E_{\alpha +\beta } & \text{if~}\alpha ,\beta ,\alpha +\beta \in \Phi
_{8},%
\end{array}
\label{crc}
\end{equation}
where $\varepsilon $ is Kac's asymmetry function, see section \ref{seo}.

\item The commutation relations in $\eno$ are the same as those for the
central extended loop algebra of $\eo$ plus derivations:\newline
\bea{l}
[t^m\otimes X \oplus \lam K \oplus\mu\ d, t^n\otimes Y \oplus \lam_1 K \oplus\nu \ d] =\\ (t^{m+n}\otimes [X,Y] -m\nu t^m\otimes X + n\mu t^n \otimes Y)\oplus m \delta_{m,-n}(X | Y)\ K
\eea
with the following correspondence%
\begin{equation}
\begin{array}{lll}
t^{m}\otimes E_{\alpha } & \leftrightarrow & x_{\alpha +m\delta }; \\
t^{m}\otimes H_{\alpha } & \leftrightarrow & x_{m\delta }^{\alpha }~\text{if~%
}m\neq 0; \\
t^{0}\otimes H_{\alpha } & \leftrightarrow & h_{\alpha }; \\
K,~d & \leftrightarrow & h_{\delta },~h_{\rho },~~\text{(see~\ref{Kd})}%
\end{array}
\label{cla}
\end{equation}
and with the invariant non-degenerate symmetric bilinear form $\left( \cdot
|\cdot \right) $ defined by:%
\begin{equation}
\left( X|Y\right) :=\left\{
\begin{array}{lll}
\left( \alpha ,\beta \right) & \text{if~}X=H_{\alpha }, & Y=H_{\beta }; \\
0 & \text{if~}X=H_{\alpha }, & Y=E_{\beta }; \\
-\delta _{\alpha ,-\beta } & \text{if~}X=E_{\alpha }, & Y=E_{\beta }.%
\end{array}%
\right.
\end{equation}
\end{enumerate}

For $\as,\bb\in \Phi _{8}$ (roots of $\eo$) and the letter $h$ referring to
a Cartan generator, the commutation relations, with no reference to the loop
algebra, are:

\bea{rcll}\label{creop}
\left[ h, h^\prime\right] &=& 0 \\
\left[ h_\as, x^\bb_{m\delta}\right] &=& 0 & m\ne 0\\
\left[ h_\delta, x^\as_{m\delta}\right] &=& 0 & m\ne 0\\
\left[ h_\rho, x^\as_{m\delta}\right] &=& m x^\as_{m\delta} & m\ne 0\\
\left[ h_\bb, x_{\as+m\delta}\right] &=& (\bb,\as) x_{\as+m\delta} \\
\left[ h_\delta, x_{\as+m\delta}\right] &=& 0 \\
\left[ h_\rho, x_{\as+m\delta}\right] &=& m x_{\as+m\delta} \\
\left[ x^\as_{m\delta}, x^\bb_{n\delta}\right] &=& m \delta_{m,-n} (\as,\bb) h_\delta& m,n\ne 0\\
\left[ x^\as_{m\delta}, x_{\bb+n\delta}\right] &=& (\as,\bb) x_{\bb+(m+n)\delta} & m\ne 0
\eea
\be\label{creo}
\left[ x_{\as+m\delta}, x_{\bb+n\delta}\right] = \left\{
\begin{array}{ll}
0 & \text{if } \as+\bb\not \in \Phi_8\cup \{0\}\\
\ep(\as,\bb)x_{\as+\bb+(m+n)\delta} &\text{if } \as+\bb \in \Phi_8\\
-x^\as_{(m+n)\delta} & \text{if } \as+\bb=0 \text{ and } m+n\ne 0\\
-h_{\as+m\delta} & \text{if } \as+\bb=0 \text{ and } m+n = 0
\end{array} \right.
\ee
\begin{remark} The commutation relations of $\eno$ are essentially determined by those of $\eo$, whose main ingredient for explicit calculations is the asymmetry function.
\end{remark}
\begin{remark}
The second correspondence in \eqref{cla} shows why the so called {\it imaginary roots} $m\delta$ are 8-fold degenerate: the space of generators associated to each root $m\delta$ is indeed an 8-dimensional space isomorphic to the span of $\{h_i,\ i=1,...,8\}$, namely the Cartan subalgebra of $\eo$.
\end{remark}

\subsection{$\eno$ in a $1+1$ dimensional toy model}

The explicit construction presented here, using a realization of the roots
in terms of the orthonormal basis $\{k_{i}\ ,\ i=-1,0,1,...,8\}$ of $\rr%
^{9,1}$, suggests to let the coordinates $k_{1},...,k_{8}$ relate to
charge/spin degrees of freedom, and to interpret the coordinates $\km,k_{0}$
as 2-momentum coordinates with Lorentzian signature.

A crucial step in our model, describing a Universe that expands from an
initial quantum state, is to restrict the particles forming that state,
hence their interactions, to lie in the subalgebra $\np$ of the triangular
decomposition \eqref{tridec} of $\eno$ (the reason will be explained in item %
\ref{posen} below). The restriction to $\np$ has the following consequences:

\benar{TM.}
\item\label{item1}The only commutation relations within \eqref{creop}
and \eqref{creo} occurring in $\np$ are, for $\as,\bb\in \Phi _{8}$: %
\bea{lll}\label{creno} \left[ x^\as_{m\delta}, x^\bb_{n\delta}\right] &= 0
&m,n\ne 0\\
\left[ x^\as_{m\delta}, x_{\bb+n\delta}\right] &= (\as,\bb) x_{\bb+(m+n)\delta} &m> 0, \ n\ge 0\\
\left[ x_{\as+m\delta}, x_{\bb+n\delta}\right] &= \left\{ \ba{l} 0\\
\ep(\as,\bb)x_{\as+\bb+(m+n)\delta}\\
-x^\as_{(m+n)\delta} \ea \right. & \ba{l} \text{if } \as+\bb\not \in
\Phi_8\cup \{0\}\\
\text{if } \as+\bb \in \Phi_8\\
\text{if } \as+\bb=0 \ea
\eea
We remark that $\as+\bb=0$ in the last commutation relation implies that
either $\as$ or $\bb$ is negative, hence $m+n\ne 0$, being $m,n\ge 0$ and in
particular $m>0$ in $x_{\as+m\delta}$ whenever $\as$ is a negative root of $%
\eo$.

\item The roots involved in the interactions are not only real; we
have for instance that $\as_{0},-\as_{9}$ are positive roots and $(\as_{0},-%
\as_{9})=-2$, therefore, by Proposition 5.1 of \cite{kac}, at least $\as%
_{0}+(-\as_{9})=\delta $ and $\as_{0}+2(-\as_{9})$ are roots; the outgoing
generator $x_{\delta }^{\as_{9}}$ of the interaction between $x_{\as_{0}}=x_{%
\as_{9}+\delta }$ and $x_{-\as_{9}}$ is similar to the Cartan generator $h_{%
\as_{9}}$, except that it carries a momentum $\delta =\um(\km+k_{0})$. This
yields the interesting consideration that \textit{neutral radiation fields
like the photon are not associated to Cartan elements of infinite
dimensional Kac-Moody algebras but to imaginary roots}, a feature that is
not present in Yang-Mills theories. It is also worth noticing that the
second equation in \eqref{creno} implies that the neutral radiation field
keeps memory of the particle-antiparticle pair which produced it,
represented in the equation by the root $\as$.

\item\label{posen} The fact that all particles are in $\np$ ensures
that their energy is always positive, even though they may be related to
both positive and negative roots of $\eo$, as revealed by the fact that $\as%
_{9}$ is the negative root of lowest height. In other words we get both
particles and antiparticles, and all of them do have positive energy.

\item The momenta given to each particle by the interactions are
lightlike. Energy momentum is conserved because the outgoing particle in an
elementary interaction is associated to a root which is the sum of the roots
of the incoming particles. \textit{All particles are massless}, since
momenta add up in the unique spatial direction.

\item We give fermionic particles helicity $1/2$.\\

So far everything runs smooth and seems physically plausible, but:\\

\item The initial quantum state of the 2-dimensional toy model under
consideration is to be a superposition of states with momenta in opposite
space directions and opposite helicity.

\item \label{td} Since $\sdel:=\frac{1}{2}(\km-k_{0})$ is not a root
of $\eno,$ we need to introduce the \textit{auxiliary roots}%
\begin{equation}
\as+m\sdel
\end{equation}%
yielding the needed superposition of momenta.\newline
We will use the notation:%
\begin{eqnarray}
\tilde{p} &:=&m\tilde{\delta}~~\text{for~}p=m\delta ;  \notag \\
\tilde{x}_{\alpha +p} &:=&x_{\alpha +p}+\eta _{\alpha }x_{\alpha +\tilde{p}%
},~~\text{where~}\eta _{\alpha }^{4}=1; \\
\tilde{x}_{p}^{\alpha } &:=&x_{p}^{\alpha }+\eta _{\alpha }x_{\tilde{p}%
}^{\alpha },  \notag
\end{eqnarray}
so that $\{x_{\as+p}\rightarrow \eta _{\as}x_{\as+\pt},\ x_{p}^{\as%
}\rightarrow x_{\pt}^{\as}\}$ is an isomorphism $\eno\rightarrow \tilde{\eno}
$ and $x_{\as}\rightarrow \eta _{\as}x_{\as}$ is an $\eo$ involution, up to
a sign. The coefficient $\eta _{\as}\in \{\pm 1,\pm \iu\}$ has been
introduced to have the freedom to vary it depending on the spin of the
generator related to $\as$.

The commutation relations \eqref{creno} become, for $\as,\bb\in \Phi_8$ and $%
p_1,p_2$ linear combination with positive integer coefficients of $\delta,%
\sdel$:
\bea{lll}\label{crenot}
\left[ x^\as_{p_1}, x^\bb_{p_2}\right] &= 0 &p_1,p_2\ne 0\\
\left[ x^\as_{p_1}, x_{\bb+p_2}\right] &= (\as,\bb) x_{\bb+p_1+p_2} &p_1\ne 0\\
\left[ x_{\as+p_1}, x_{\bb+p_2}\right] &= \left\{
\ba{l}
0\\
\ep(\as,\bb)x_{\as+\bb+ p_1+p_2}\\
-x^\as_{p_1+p_2}
\ea \right.
&
\ba{l}
\text{if } \as+\bb\not \in \Phi_8\cup \{0\}\\
\text{if } \as+\bb \in \Phi_8\\
\text{if } \as+\bb=0
\ea
\eea
$\as+\bb=0$ in the last commutation relation implies $p_1+p_2
\ne 0$, as remarked in item \ref{item1}. Moreover $(p_1+p_2)^2\ge 0$ and we
still have positive energy associated to all particles. It is no longer
true, however, that particles are necessarily massless, as we immediately
realize by the fact that $\delta+\sdel=k_{-1}$ represents a mass at rest. We
also notice that the product is to be antisymmetric, therefore
\be
x^{-\as}_{p} = - x^{\as}_{p} \text{ and } \left[ x^\as_{p_1}, x_{\bb+p_2}\right] = -\left[ x_{\bb+p_2}, x^\as_{p_1} \right]
\ee
Moreover, for consistency
\be\label{apb}
x^{\as+\bb}_{p}=x^{\as}_{p} + x^{\bb}_{p}
\ee
We will prove in the forthcoming paper \cite{mym2}, as a particular case of
a more general statement, that the algebra so defined is a Lie algebra.

\item \label{sd} We lack two spatial dimensions. This suggests a
further extension to $\edo$ or $\dedo$, or to analogous Borcherds (or
generalized Kac-Moody) algebras, as we will investigate in the next sections.

\item Our toy model still lacks three features, which urges to a
further extension of the algebra (investigated in the companion paper \cite%
{mym2}):
\ben

\item \textit{locality}, i.e. spacetime related multiplication rules that
immerse the algebra into a vertex-type algebra;

\item \textit{space expansion} within the vertex algebra;

\item \textit{Pauli exclusion principle}, that, as we will see, requires an
extension to Lie superalgebra.
\een
\een

The above considerations are hinting to the fact that the extension of $\eo$
to Kac-Moody, or even beyond, to Generalized Kac-Moody (Borcherds), algebras
is very appealing to particle physics, not only to 2-dimensional conformal
field theory, \cite{CFT}.

\subsection{$\edo$ and $\dedo$}

The Dynkin diagram of $\edo$ is shown in \eqref{dynkedo}. We use the same
indices for the simple roots $\rmu,\ropp,\rop,\as_0,\as_1,...,\as_8$ and the
orthonormal basis vectors $\km,\kopp,\kop,k_0,k_1,...,k_8$ of the Lorentzian
space $\rr^{11,1}$, with $(k_i,k_i)=1$, for $i = 0^{\prime \prime
},0^{\prime },0,1,...,8$, and $(\km,\km)=-1$. \newline

It is here worth presenting a different choice for the set of simple roots,
with respect to \eqref{deltaeno}, in which the $\eno$ simple roots are all
fermionic:
\bea{lcl}
\alpha _{-1} &=&-\frac{\sqrt{2}}{2}k_{0^{\prime \prime }}+\frac{\sqrt{6}}{2}%
k_{0^{\prime }}+\sqrt{3}\delta ;\vspace{4pt}   \\
\alpha _{0^{\prime \prime }} &=&\sqrt{2}k_{0^{\prime \prime }}+\delta ;\vspace{4pt}
 \\
\alpha _{0^{\prime }} &=&\sqrt{2}k_{0^{\prime }}-2k_{0}+2\delta ;\vspace{4pt}   \\
\alpha _{0} &=&\frac{1}{2}
(-k_{1}+k_{2}+k_{3}-k_{4}-k_{5}-k_{6}-k_{7}-k_{8})+\delta
;\vspace{4pt}   \\
\alpha _{1} &=&\frac{1}{2}(k_{1}-k_{2}-k_{3}-k_{4}+k_{5}+k_{6}+k_{7}-k_{8})
;\vspace{4pt}  \label{rsr12} \\
\alpha _{2} &=&\frac{1}{2}(k_{1}-k_{2}+k_{3}+k_{4}-k_{5}-k_{6}-k_{7}+k_{8})
;\vspace{4pt}   \\
\alpha _{3} &=&\frac{1}{2}(-k_{1}+k_{2}-k_{3}+k_{4}-k_{5}+k_{6}+k_{7}-k_{8})
;\vspace{4pt}   \\
\alpha _{4} &=&\frac{1}{2}(k_{1}+k_{2}-k_{3}-k_{4}+k_{5}-k_{6}-k_{7}+k_{8})
;\vspace{4pt}   \\
\alpha _{5} &=&\frac{1}{2}(-k_{1}-k_{2}+k_{3}+k_{4}+k_{5}-k_{6}+k_{7}-k_{8});\vspace{4pt}
 \\
\alpha _{6} &=&\frac{1}{2}(k_{1}+k_{2}+k_{3}-k_{4}-k_{5}+k_{6}+k_{7}+k_{8})
;\vspace{4pt}   \\
\alpha _{7} &=&\frac{1}{2}(-k_{1}-k_{2}-k_{3}-k_{4}-k_{5}+k_{6}-k_{7}+k_{8});\vspace{4pt}
 \\
\alpha _{8} &=&\frac{1}{2}(k_{1}+k_{2}+k_{3}+k_{4}+k_{5}+k_{6}-k_{7}-k_{8}),
\eea
where $\delta =\frac{1}{2}(k_{0}+\km)$.\newline
The corresponding Cartan matrix is the Gram matrix of the $\edo$ lattice in $%
\rr^{11,1}$ which is not unimodular. We interpret the coordinates $(\kopp,%
\kop,k_{0},\km)$ as 4-momentum coordinates with Lorentzian signature.

\medskip

Similar arguments hold for the Kac-Moody algebra $\dedo$, whose Dynkin
diagram is:
\be
\setlength{\unitlength}{3pt}
\begin{picture}(85,15)
\linethickness{0.3mm}
\multiput(2,5)(8,0){10}{\circle{2}}
\put(58,13){\circle{2}}
\put(10,13){\circle{2}}
\put(9.9,6){\line(0,1){6}}
\put(57.9,6){\line(0,1){6}}
\multiput(3,5)(8,0){9}{\line(1,0){6}}
\put(0,0){$\ropp$}
\put(8,0){$\rop$}
\put(17,0){$\as_0$}
\put(25,0){$\as_1$}
\put(33,0){$\as_2$}
\put(41,0){$\as_3$}
\put(49,0){$\as_4$}
\put(57,0){$\as_5$}
\put(65,0){$\as_7$}
\put(73,0){$\as_8$}
\put(60,12){$\as_6$}
\put(12,12){$\rmu$}
\end{picture}
\label{dynkdedo}
\ee

This is the extension of $\eo$ through the orthogonal Lie algebra $\mathbf{d}%
_4$.

A possible set of simple roots, in the orthonormal basis of the Lorentzian
space $\rr^{11,1}$, is:%
\bea{lcl}
\alpha _{-1} &=&\sqrt{2}k_{0^{\prime }}+3\delta ;  \vspace{4pt} \\
\alpha _{0^{\prime \prime }} &=&\sqrt{2}k_{0^{\prime \prime }}+\delta ;
\vspace{4pt} \\
\alpha _{0^{\prime }} &=&\sqrt{2}k_{0^{\prime }}-2k_{0}+2\delta ;  \vspace{4pt} \\
\alpha _{0} &=&\frac{1}{2}%
(-k_{1}+k_{2}+k_{3}-k_{4}-k_{5}-k_{6}-k_{7}-k_{8})+\delta;  \vspace{4pt} \\
\alpha _{1} &=&\frac{1}{2}(k_{1}-k_{2}-k_{3}-k_{4}+k_{5}+k_{6}+k_{7}-k_{8})
;  \label{rsr12p} \\
\alpha _{2} &=&\frac{1}{2}(k_{1}-k_{2}+k_{3}+k_{4}-k_{5}-k_{6}-k_{7}+k_{8})
;  \vspace{4pt} \\
\alpha _{3} &=&\frac{1}{2}(-k_{1}+k_{2}-k_{3}+k_{4}-k_{5}+k_{6}+k_{7}-k_{8})
;  \vspace{4pt} \\
\alpha _{4} &=&\frac{1}{2}(k_{1}+k_{2}-k_{3}-k_{4}+k_{5}-k_{6}-k_{7}+k_{8})
;  \vspace{4pt} \\
\alpha _{5} &=&\frac{1}{2}(-k_{1}-k_{2}+k_{3}+k_{4}+k_{5}-k_{6}+k_{7}-k_{8})
;  \vspace{4pt} \\
\alpha _{6} &=&\frac{1}{2}(k_{1}+k_{2}+k_{3}-k_{4}-k_{5}+k_{6}+k_{7}+k_{8})
;  \vspace{4pt} \\
\alpha _{7} &=&\frac{1}{2}(-k_{1}-k_{2}-k_{3}-k_{4}-k_{5}+k_{6}-k_{7}+k_{8});
\vspace{4pt} \\
\alpha _{8} &=&\frac{1}{2}(k_{1}+k_{2}+k_{3}+k_{4}+k_{5}+k_{6}-k_{7}-k_{8}),
\eea
where $\delta :=\frac{1}{2}(k_{0}+\km)$. One can realize at a glance that
this is the same set of simple roots of $\edo$ except for the root $\rmu$.
We think that the possibility to discriminate between $\edo$ and $\dedo$ on
physical grounds can only arising when performing explicit computer
calculations, which for the case of $\dedo$ should undergo major
simplifications, due to the presence of only one irrational number, $\sqrt{2}
$.

\section{Beyond Kac-Moody\label{diff}}

Several difficulties in proceeding with our program arise with Kac-Moody
algebras:

\benar{P.}
\item \label{prob1} As already mentioned above, the presence of
irrational numbers in the definition of momentum variables can be very
annoying in computer calculations, and it is also quite unnatural in an
algebra based on integer numbers, both in the roots and in the structure
constants.

\item \label{prob2} The need of roots, as in item \ref{td} for $\eno$,
with opposite helicity and opposite signs of the 3-momentum components $\kopp%
,\kop,k_{0}$ complicates the algebra very much, since \textit{they cannot in
pairs be roots of the Kac-Moody algebras}. In the case of $\eno$, this
problem can be overcome by enlarging the explicit commutation relations to
include consistently the new roots, but in the case of $\edo$ or $\dedo$ one
needs to further modify the Serre relations, which does not seem an easy
task to us.

\item \label{prob3} However, the most important issue comes from
physics: in $\edo$ and $\dedo$ three simple roots (namely, $\rmu$, $\ropp$
and $\rop$) have \textit{tachyon-like} momenta, due to their positive norm.
The interpretation of such tachyonic momenta, as well as the investigation
of their impact on the interactions among charged particles, is beyond the
scope of the present paper; computer calculations, starting from an initial
state, may reveal the scenario that tachyonic simple roots may yield to.
\een

It is our opinion that these are good motivations for focusing our
investigation on Generalized Kac-Moody (Borcherds) algebras, where two of
the three difficulties listed above disappear.

\section{Borcherds $\bdo$}

Borcherds algebras are a generalization of Kac-Moody algebras obtained by
releasing the condition on the diagonal elements of the Cartan matrix, which
are then allowed to be non-positive, as well as by restricting the Serre
relations to the generators associated to positive norm simple roots, \cite%
{borc1}, \cite{borc2}.\newline
A generalized Kac-Moody (or Borcherds) algebra $\borc$ is constructed as
follows.

Let $H$ be a real vector space with a symmetric bilinear inner product $%
(\cdot ,\cdot )$, and with elements $h_{i}$ indexed by a countable set $%
\ical
$, such that $(h_{i},h_{j})\leq 0$ if $i\neq j$ and $%
2(h_{i},h_{j})/(h_{i},h_{i})$ is an integer if $(h_{i},h_{i})$ is positive.
The matrix $A$ with entries $a_{ij}:=(h_{i},h_{j})$ is called the
symmetrized Cartan matrix of $\borc$.\newline
The generalized Kac-Moody (or Borcherds) algebra $\borc$ associated to $A$
is defined to be the Lie algebra generated by $H$ and elements $e_{i}$ and $%
f_{i}$, for $i\in \ical$, with the following relations:

\begin{enumerate}
\item The (injective) image of $H$ in $\borc$ is commutative.

\item If $h$ is in $H$, then $[h,e_{i}]=(h,h_{i})e_{i}$ and $%
[h,f_{i}]=-(h,h_{i})f_{i}$.

\item $[e_{i},f_{j}]=\delta _{ij}h_{i}$.

\item If $a_{ii}>0$ and $i\neq j$, then ad$(e_{i})^{n}\ e_{j}=\text{ad}%
(f_{i})^{n}\ f_{j}=0$, where $n=1-2a_{ij}/a_{ii}$.

\item If $a_{ij}=0$, then $[e_{i},e_{j}]=[f_{i},f_{j}]=0$.
\end{enumerate}

If $a_{ii}>0$ for all $i\in \ical$, then $\borc$ is the Kac-Moody algebra
with Cartan matrix $A$. In general, $\borc$ has almost all the properties of
a Kac-Moody algebra, the only major difference being that $\borc$ is allowed
to have \textit{imaginary simple roots}.

The root lattice $\Ll$ is the free Abelian group generated by elements $\as%
_{i}$ for $i\in \ical$, called \textit{simple roots}, and $\Ll$ has a
real-valued bilinear form defined by $(\as_{i},\as_{j})=a_{ij}$ . The Lie
algebra $\borc$ is then graded by $\Ll$ with $H$ in degree $0$, $e_{i}$
(resp $f_{i}$) in degree $\as_{i}$ (resp. $-\as_{i}$). A root is a nonzero
element $\as$ of $\Ll$ such that there are elements of $\borc$ of degree $%
\as
$. A root $r$ is called \textit{real} if $(r,r)>0$, otherwise it is called
\textit{imaginary}. A root $r$ is positive if it is a sum of simple roots,
and negative if $-r$ is positive. Notice that \textit{every root is either
positive or negative}, \cite{borc1}.

We build the following symmetrized Cartan matrix for a Borcherds algebra of
rank 12, that we denote by $\bdo$:%
\begin{equation}
\left(
\begin{array}{cccccccccccc}
-1 & -1 & -1 & -1 & 0 & 0 & 0 & 0 & 0 & 0 & 0 & 0 \\
-1 & 0 & -1 & -1 & 0 & 0 & 0 & 0 & 0 & 0 & 0 & 0 \\
-1 & -1 & 0 & -1 & 0 & 0 & 0 & 0 & 0 & 0 & 0 & 0 \\
-1 & -1 & -1 & 2 & -1 & 0 & 0 & 0 & 0 & 0 & 0 & 0 \\
0 & 0 & 0 & -1 & 2 & -1 & 0 & 0 & 0 & 0 & 0 & 0 \\
0 & 0 & 0 & 0 & -1 & 2 & -1 & 0 & 0 & 0 & 0 & 0 \\
0 & 0 & 0 & 0 & 0 & -1 & 2 & -1 & 0 & 0 & 0 & 0 \\
0 & 0 & 0 & 0 & 0 & 0 & -1 & 2 & -1 & 0 & 0 & 0 \\
0 & 0 & 0 & 0 & 0 & 0 & 0 & -1 & 2 & -1 & -1 & 0 \\
0 & 0 & 0 & 0 & 0 & 0 & 0 & 0 & -1 & 2 & 0 & 0 \\
0 & 0 & 0 & 0 & 0 & 0 & 0 & 0 & -1 & 0 & 2 & -1 \\
0 & 0 & 0 & 0 & 0 & 0 & 0 & 0 & 0 & 0 & -1 & 2%
\end{array}%
\right) .  \label{cm12}
\end{equation}

Notice that, for $\delta :=\as_{0}+2\as_{1}+3\as_{2}+4\as_{3}+5\as_{4}+6\as%
_{5}+3\as_{6}+4\as_{7}+2\as_{8}$, a 4-momentum vector can be written as%
\begin{equation}
p:=p_{0}\rmu+p_{1}(\ropp-\rmu)+p_{2}(\rop-\rmu)+p_{3}(\delta -\rmu).
\label{fmom}
\end{equation}
Using the Cartan Matrix \eqref{cm12} we get indeed:%
\begin{equation}
\begin{array}{l}
(\rmu,\rmu)=-1\newline
; \\
(\ropp-\rmu,\ropp-\rmu)=(\rop-\rmu,\rop-\rmu)=(\delta -\rmu,\delta -\rmu)=1%
\newline
; \\
(\rmu,\ropp-\rmu)=(\rmu,\rop-\rmu)=(\rmu,\delta -\rmu)=0\newline
; \\
(\ropp-\rmu,\rop-\rmu)=(\ropp-\rmu,\delta -\rmu)=(\rop-\rmu,\delta -\rmu)=0,%
\end{array}%
\end{equation}
hence the Lorentzian scalar product:%
\begin{equation}
(p,p^{\prime })=-p_{0}p_{0}^{\prime }+p_{1}p_{1}^{\prime
}+p_{2}p_{2}^{\prime }+p_{3}p_{3}^{\prime }.
\end{equation}

Let us restrict to positive roots $r=\sum_{\ical}\lam_{i}\as_{i}$, $\ical%
=\{-1,0^{\prime \prime },0^{\prime },0,...,8\}$, with $\lam_{i}\in \nn\cup
\{0\}$, and let us denote by $\Bp$ the corresponding subalgebra of $\bdo$.
The physical motivation for restricting to $\Bp$ is that, given a positive
root $r=\sum_{\ical}\lam_{i}\as_{i}$, its 4-momentum is%
\begin{equation}
p=(p_{0},p_{1},p_{2},p_{3})=(\lam_{-1}+\lam_{0\prime \prime }+\lam_{0\prime
}+\lam_{0},\lam_{0\prime \prime },\lam_{0\prime },\lam_{0}),  \label{palfa}
\end{equation}
with $\lam_{-1},\lam_{0^{\prime \prime }},\lam_{0^{\prime }},\lam_{0}\geq 0$%
, implying $m^{2}:=-p^{2}\geq 0$, namely $p$ either lightlike or timelike.
In particular:
\bea{ll}
p^2 &= -(\lam_{-1}^2+2\lam_{-1}\sum_{i\ne -1}{\lam_i}+\sum_{i\ne j,\, i,j\ne-1}{\lam_i\lam_j})\ , \ i,j\in\{-1, 0'',0',0\}\vspace{1em}\\
& \left\{
\ba{ll}
=0 &\text{ if $\lam_{-1}=0$ and at most one $\lambda_i\ne 0\, , i\ne-1$}\\
=-1 &\text{ if $\lam_{-1}=1$ and all $\lambda_i= 0\, , i\ne-1$}\\
\le -2 &\text{ otherwise}
\ea
\right.
\eea
\begin{remark}
	Notice that the mass of a particle cannot be arbitrary small, since there is a lower limit $m\ge 1$.
\end{remark}
For $r = \sum_\ical \lam_i \as_i$, $\ical = \{ -1, 0^{\prime \prime
},0^{\prime },0,...,8\}$, with $\lam_i\in \nn\cup\{0\}$ we introduce the
notation
\bea{ll}
r=\as+p &p:=\lam_{-1}\as_{-1}+\lam_{0''}\as_{0''}+\lam_{0'}\as_{0'}+\lam_0\delta\text{ (see \eqref{fmom} and \eqref{palfa})}\\ &\as:=\lam_0\as_9+\lam_1\as_1+...+\lam_8\as_8\\
&\php{\as}=(\lam_1-2\lam_0)\as_1+(\lam_2-3\lam_0)\as_2+(\lam_3-4\lam_0)\as_3\\
&\php{\as}+(\lam_4-5\lam_0)\as_4+(\lam_5-6\lam_0)\as_5+(\lam_6-3\lam_0)\as_6\\
&\php{\as}+(\lam_7-4\lam_0)\as_7+(\lam_8-2\lam_0)\as_8
\eea
Thus, $\as$ is in the lattice $\Lleo$ of $\eo$, and a precise physical
meaning is assigned to positive real and imaginary roots when $\as\in \Lleo %
\backslash\{0\}$:

\begin{proposition}\label{pmass}
	A generator in $\Bp$, associated to a positive root $r=\as+p$, with $\as\in\Lleo \backslash\{0\}$ and momentum $p\ne 0$, is \textit{massive} if and only if $\as+p$ is an imaginary root; it is \textit{massless} if and only if $r$ is real, in which case it is a positive real root of $\eno\subset\bdo$.
\end{proposition}
\begin{proof}
	The proof consists of the following steps:
	\ben
	\item from Proposition 2.1. of \cite{borc1}, it holds that every positive root $r=\alpha+p$ is conjugate under the Weyl group to a root $r_0=\alpha^\prime+p^\prime$, such that either $r_0$ is a simple real root $\as_i$, $i\in \{0,1,...,8\}$, or it is a positive root in the \textit{Weyl chamber} (namely $(r_0,\as_i)\le 0$ for all simple roots $\as_i$);
	\item since $r$ and $r_0$ are conjugate under the Weyl group, then $(r,r)=(r_0,r_0)$;
	\item if $r_0$ is a real simple root, then it is a root of $\eno$ and ${p^\prime}^2=0$; $r_0$ is real and so is $r$. Since the Weyl group is generated by the reflections $\rho-(\rho,\as_i)\as_i$, where the $\as_i$ simple roots are real, hence $\as_i\in\eno\subset\bdo$, it coincides with the Weyl group of $\eno$. By applying to $r_0$ Weyl reflections we stay within $\eno$, since every Kac-Moody algebra is invariant under the Weyl group; therefore, $r$ is a real root of $\eno$, namely $r=\alpha+m\delta$, $\alpha\in \eo$ and $m\delta$ is lightlike;
	\item if $r_0$ is in the Weyl chamber, then $(r,r)=(r_0,r_0)=\sum \lam_i (r_0,\as_i)\le 0$, $i\in\ical$, since all $\lam_i$ are positive being $r_0$ is a positive root; thus, $r$ is imaginary;
	\item \label{mass} since $(r,r)=\as^2 +p^2\le 0$ with $\as^2\ge 2$, then $m^2\ge \as^2\ge 2$, and the particle associated to $r$ is massive.
	\een
\end{proof}

\begin{remark}
	In the massive case, the lower limit of the mass grows with the norm of $\as$: if $\as\in \Lleo\backslash\{0\}$ is not a root of $\Phi_8$, then the mass is certainly bigger than the lower mass a particle corresponding to a root of $\Phi_8$ may have.
	We also notice that charged massless particles ($\as\ne 0$ in the root $\as+p$) are quite peculiar, since their momentum can only be in 1 direction. The photon is not in this class, since it has $\as=0$, but the (non-virtual) gluons are. A non-virtual photon can be produced in a decay process, \emph{\cite{mym2}}.
\end{remark}
\begin{remark}
	We emphasize that two of the three problems listed in section \sref{diff} about Kac-Moody algebras vanish in the Borcherds algebra $\bdo$. These are obviously \emph{\ref{prob1}} and \emph{\ref{prob3}}. But \emph{\ref{prob2}} still remains, \emph{\cite{mym2}}.
\end{remark}

The companion paper \cite{mym2}, based on the treatment and considerations
of this paper, will focus on a particular rank-12 algebra $\ggu$, in order
to build a model for quantum gravity. In particular, we will turn $\ggu$
into a Lie superalgebra, and we will discuss scattering processes and decays.

\begin{table}[tbp]
\begin{equation*}
{\scriptsize {\
\begin{array}{|cl|c|c|c|}
\hline
\hspace{2.9cm} roots &  & q_{e.m.} & \# & \{r_c,s_c\} \\ \hline
\pm (k_1 - k_2) &  & 0 & 2 & \pm \{2,0\} \\
\pm (k_2 - k_3) &  & 0 & 2 & \pm \{-1,3\} \\
\pm (k_1 - k_3) &  & 0 & 2 & \pm \{1,3\} \\ \hline
&  &  &  &  \\
\pm k_i \pm k_j & 5 \le i < j \le 8 & 0 & 24 &  \\
k_4 \pm k_i & 5 \le i \le 8 & 1 & 8 &  \\
-k_4 \pm k_i & 5 \le i \le 7 & -1 & 8 &  \\
&  &  &  & \{0,0\} \\
\frac{1}{2} (\pm (k_1 + k_2 + k_3 + k_4) \pm ... \pm k_8) & \text{even \# of
+} & 0 & 16 &  \\
\frac{1}{2} (- (k_1 + k_2 + k_3) + k_4 \pm ... \pm k_8) & \text{even \# of +}
& 1 & 8 &  \\
\frac{1}{2} (+ (k_1 + k_2 + k_3) - k_4 \pm ... \pm k_8) & \text{even \# of +}
& -1 & 8 &  \\
&  &  &  &  \\ \hline
&  &  &  &  \\
-k_2 - k_3 \ , \ k_1 + k_4 &  & {\scriptstyle 2/3} & 2 &  \\
k_1 - k_4 &  & -{\scriptstyle 4/3} & 1 &  \\
k_1 \pm k_i & i=5,... , 8 & -{\scriptstyle 1/3} & 8 & \{1,1\} \\
\frac{1}{2} (k_1 - k_2 - k_3 + k_4 \pm ... \pm k_8) & \text{even \# of +} & {%
\scriptstyle 2/3} & 8 &  \\
\frac{1}{2} (k_1 - k_2 - k_3 - k_4 \pm ... \pm k_8) & \text{even \# of +} & -%
{\scriptstyle 1/3} & 8 &  \\
&  &  &  &  \\ \hline
&  &  &  &  \\
+k_2 + k_3 \ , \ -k_1 - k_4 &  & -{\scriptstyle 2/3} & 2 &  \\
-k_1 + k_4 &  & {\scriptstyle 4/3} & 1 &  \\
-k_1 \pm k_i & i=5,... , 8 & {\scriptstyle 1/3} & 8 & \{-1,-1\} \\
\frac{1}{2} (-k_1 + k_2 + k_3 - k_4 \pm ... \pm k_8) & \text{even \# of +} &
-{\scriptstyle 2/3} & 8 &  \\
\frac{1}{2} (-k_1 + k_2 + k_3 + k_4 \pm ... \pm k_8) & \text{even \# of +} &
{\scriptstyle 1/3} & 8 &  \\
&  &  &  &  \\ \hline
&  &  &  &  \\
-k_1 - k_3 \ , \ k_2 + k_4 &  & {\scriptstyle 2/3} & 2 &  \\
k_2 - k_4 &  & -{\scriptstyle 4/3} & 1 &  \\
k_2 \pm k_i & i=5,... , 8 & -{\scriptstyle 1/3} & 8 & \{-1,1\} \\
\frac{1}{2} (-k_1 + k_2 - k_3 + k_4 \pm ... \pm k_8) & \text{even \# of +} &
{\scriptstyle 2/3} & 8 &  \\
\frac{1}{2} (-k_1 + k_2 - k_3 - k_4 \pm ... \pm k_8) & \text{even \# of +} &
-{\scriptstyle 1/3} & 8 &  \\
&  &  &  &  \\ \hline
&  &  &  &  \\
k_1 + k_3 \ , \ -k_2 - k_4 &  & -{\scriptstyle 2/3} & 2 &  \\
-k_2 + k_4 &  & {\scriptstyle 4/3} & 1 &  \\
-k_2 \pm k_i & i=5,... , 8 & {\scriptstyle 1/3} & 8 & \{1,-1\} \\
\frac{1}{2} (k_1 - k_2 + k_3 - k_4 \pm ... \pm k_8) & \text{even \# of +} & -%
{\scriptstyle 2/3} & 8 &  \\
\frac{1}{2} (k_1 - k_2 + k_3 + k_4 \pm ... \pm k_8) & \text{even \# of +} & {%
\scriptstyle 1/3} & 8 &  \\
&  &  &  &  \\ \hline
&  &  &  &  \\
-k_1 - k_2 \ , \ k_3 + k_4 &  & {\scriptstyle 2/3} & 2 &  \\
k_3 - k_4 &  & -{\scriptstyle 4/3} & 1 &  \\
k_3 \pm k_i & i=5,... , 8 & -{\scriptstyle 1/3} & 8 & \{0,-2\} \\
\frac{1}{2} (-k_1 - k_2 + k_3 + k_4 \pm ... \pm k_8) & \text{even \# of +} &
{\scriptstyle 2/3} & 8 &  \\
\frac{1}{2} (-k_1 - k_2 + k_3 - k_4 \pm ... \pm k_8) & \text{even \# of +} &
-{\scriptstyle 1/3} & 8 &  \\
&  &  &  &  \\ \hline
&  &  &  &  \\
k_1 + k_2 \ , \ -k_3 - k_4 &  & -{\scriptstyle 2/3} & 2 &  \\
-k_3 + k_4 &  & {\scriptstyle 4/3} & 1 &  \\
-k_3 \pm k_i & i=5,... , 8 & {\scriptstyle 1/3} & 8 & \{0,2\} \\
\frac{1}{2} (k_1 + k_2 - k_3 - k_4 \pm ... \pm k_8) & \text{even \# of +} & -%
{\scriptstyle 2/3} & 8 &  \\
\frac{1}{2} (k_1 + k_2 - k_3 + k_4 \pm ... \pm k_8) & \text{even \# of +} & {%
\scriptstyle 1/3} & 8 &  \\
&  &  &  &  \\ \hline
\end{array}
}}
\end{equation*}%
\caption{The Magic Star of $\eo$; $q_{e.m.}(\as)=(\as,-%
\frac13(k_1+k_2+k_3)+k_4)$; $\php{Table 3:cc} r_c(\as)=(\as,k_1-k_2)$, $s_c(
\as)=(\as,k_1+k_2-2k_3)$.}
\label{t:magiceo}
\end{table}

\begin{table}[hp]
\begin{equation*}
{\scriptsize {\
\begin{array}{|cl|c|c|c|}
\hline
\hspace{2.9cm} roots &  & q_{e.m.} & \#\ of\ roots & \{r_f,s_f\} \\ \hline
\pm (k_4 - k_5) &  & \pm 1 & 2 & \pm \{2,0\} \\
\pm (k_5 - k_6) &  & 0 & 2 & \pm \{-1,3\} \\
\pm (k_4 - k_6) &  & \pm 1 & 2 & \pm \{1,3\} \\ \hline
&  &  &  &  \\
\pm k_7 \pm k_8 &  & 0 & 4 &  \\
\frac{1}{2} (\pm (k_1 + k_2 + k_3 + k_4+k_5+k_6) \pm k_7 \pm k_8) & \text{%
even \# of +} & 0 & 4 &  \\
&  &  &  & \{0,0\} \\
\frac{1}{2}(- k_1 - k_2 - k_3 + k_4 +k_5+k_6 \pm k_7 \pm k_8) & \text{even
\# of +} & 1 & 2 &  \\
\frac{1}{2}( k_1 + k_2 + k_3 - k_4 -k_5-k_6 \pm k_7 \pm k_8) & \text{even \#
of +} & -1 & 2 &  \\
&  &  &  &  \\ \hline
&  &  &  &  \\
-k_5 - k_6 &  & 0 & 1 &  \\
k_4 \pm k_i & i=7, 8 & 1 & 4 &  \\
&  &  &  & \{1,1\} \\
\frac{1}{2} (k_1 + k_2 + k_3 + k_4 - k_5 - k_6 \pm k_7 \pm k_8) & \text{even
\# of +} & 0 & 2 &  \\
\frac{1}{2} (-k_1 - k_2 - k_3 + k_4 - k_5 - k_6 \pm k_7 \pm k_8) & \text{%
even \# of +} & 1 & 2 &  \\
&  &  &  &  \\ \hline
&  &  &  &  \\
k_5 + k_6 &  & 0 & 1 &  \\
-k_4 \pm k_i & i=7, 8 & -1 & 4 &  \\
&  &  &  & \{-1,-1\} \\
\frac{1}{2} (k_1 + k_2 + k_3 - k_4 + k_5 + k_6 \pm k_7 \pm k_8) & \text{even
\# of +} & -1 & 2 &  \\
\frac{1}{2} (-k_1 - k_2 - k_3 - k_4 + k_5 + k_6 \pm k_7 \pm k_8) & \text{%
even \# of +} & 0 & 2 &  \\
&  &  &  &  \\ \hline
&  &  &  &  \\
-k_4 - k_6 &  & -1 & 1 &  \\
k_5 \pm k_i & i=7, 8 & 0 & 4 &  \\
&  &  &  & \{-1,1\} \\
\frac{1}{2} (k_1 + k_2 + k_3 - k_4 + k_5 - k_6 \pm k_7 \pm k_8) & \text{even
\# of +} & -1 & 2 &  \\
\frac{1}{2} (-k_1 - k_2 - k_3 - k_4 + k_5 - k_6 \pm k_7 \pm k_8) & \text{%
even \# of +} & 0 & 2 &  \\
&  &  &  &  \\ \hline
&  &  &  &  \\
k_4 + k_6 &  & 1 & 1 &  \\
-k_5 \pm k_i & i=7, 8 & 0 & 4 &  \\
&  &  &  & \{1,-1\} \\
\frac{1}{2} (k_1 + k_2 + k_3 + k_4 - k_5 + k_6 \pm k_7 \pm k_8) & \text{even
\# of +} & 0 & 2 &  \\
\frac{1}{2} (-k_1 - k_2 - k_3 + k_4 - k_5 + k_6 \pm k_7 \pm k_8) & \text{%
even \# of +} & 1 & 2 &  \\
&  &  &  &  \\ \hline
&  &  &  &  \\
-k_4 - k_5 &  & -1 & 1 &  \\
k_6 \pm k_i & i=7, 8 & 0 & 4 &  \\
&  &  &  & \{0,-2\} \\
\frac{1}{2} (k_1 + k_2 + k_3 - k_4 - k_5 + k_6 \pm k_7 \pm k_8) & \text{even
\# of +} & -1 & 2 &  \\
\frac{1}{2} (-k_1 - k_2 - k_3 - k_4 - k_5 + k_6 \pm k_7 \pm k_8) & \text{%
even \# of +} & 0 & 2 &  \\
&  &  &  &  \\ \hline
&  &  &  &  \\
k_4 + k_5 &  & 1 & 1 &  \\
-k_6 \pm k_i & i=7, 8 & 0 & 4 &  \\
&  &  &  & \{0,2\} \\
\frac{1}{2} (k_1 + k_2 + k_3 + k_4 + k_5 - k_6 \pm k_7 \pm k_8) & \text{even
\# of +} & 0 & 2 &  \\
\frac{1}{2} (-k_1 - k_2 - k_3 + k_4 + k_5 - k_6 \pm k_7 \pm k_8) & \text{%
even \# of +} & 1 & 2 &  \\
&  &  &  &  \\ \hline
\end{array}
}}
\end{equation*}%
\caption{The Magic Star of $\es$; $q_{e.m.}(\as)=(\as,-%
\frac13(k_1+k_2+k_3)+k_4)$; $\php{Table 3:cc} r_f(\as)=(\as,k_4-k_5)$, $s_f(
\as)=(\as,k_4+k_5-2k_6)$.}
\label{t:magices}
\end{table}

\begin{table}[hp]
\begin{equation*}
{\scriptsize {\
\begin{array}{|c|c|c|c|c|}
\hline
&  &  &  &  \\
roots\ \left[\ k:= k_1+k_2+k_3\ \right] & q_{e.m.} & lepton &
\sigma_z:=\frac12(\as,\rho_1+\rho_2) & \{r_f,s_f\} \\
&  &  &  &  \\ \hline
&  &  &  &  \\
\frac{1}{2} (k + k_4 - k_5 - k_6 + k_7 + k_8) & 0 & \ntm & -1/2 &  \\
\frac{1}{2} (k + k_4 - k_5 - k_6 - k_7 - k_8) & 0 & \ncm & -1/2 &  \\
&  &  &  & \{1,1\} \\
\frac{1}{2} (-k + k_4 - k_5 - k_6 + k_7 - k_8) & 1 & e^+ & -1/2 &  \\
\frac{1}{2} (-k + k_4 - k_5 - k_6 - k_7 + k_8) & 1 & \mu^+ & -1/2 &  \\
&  &  &  &  \\ \hline
&  &  &  &  \\
-\frac{1}{2} (k + k_4 - k_5 - k_6 + k_7 + k_8) & 0 & \ntp & 1/2 &  \\
-\frac{1}{2} (k + k_4 - k_5 - k_6 - k_7 - k_8) & 0 & \ncp & 1/2 &  \\
&  &  &  & \{-1,-1\} \\
-\frac{1}{2} (-k + k_4 - k_5 - k_6 + k_7 - k_8) & -1 & e^- & 1/2 &  \\
-\frac{1}{2} (-k + k_4 - k_5 - k_6 - k_7 + k_8) & -1 & \mu^- & 1/2 &  \\
&  &  &  &  \\ \hline
&  &  &  &  \\
\frac{1}{2} (k - k_4 + k_5 - k_6 + k_7 + k_8) & -1 & \tau^- & 1/2 &  \\
\frac{1}{2} (k - k_4 + k_5 - k_6 - k_7 - k_8) & -1 & \chi^- & 1/2 &  \\
&  &  &  & \{-1,1\} \\
\frac{1}{2} (-k - k_4 + k_5 - k_6 + k_7 - k_8) & 0 & \nep & 1/2 &  \\
\frac{1}{2} (-k - k_4 + k_5 - k_6 - k_7 + k_8) & 0 & \nmp & 1/2 &  \\
&  &  &  &  \\ \hline
&  &  &  &  \\
-\frac{1}{2} (k - k_4 + k_5 - k_6 + k_7 + k_8) & 1 & \tau^+ & -1/2 &  \\
-\frac{1}{2} (k - k_4 + k_5 - k_6 - k_7 - k_8) & 1 & \chi^+ & -1/2 &  \\
&  &  &  & \{1,-1\} \\
-\frac{1}{2} (-k - k_4 + k_5 - k_6 + k_7 - k_8) & 0 & \nem & -1/2 &  \\
-\frac{1}{2} (-k - k_4 + k_5 - k_6 - k_7 + k_8) & 0 & \nmm & -1/2 &  \\
&  &  &  &  \\ \hline
&  &  &  &  \\
\frac{1}{2} (k - k_4 - k_5 + k_6 + k_7 + k_8) & -1 & \tau^- & -1/2 &  \\
\frac{1}{2} (k - k_4 - k_5 + k_6 - k_7 - k_8) & -1 & \chi^- & -1/2 &  \\
&  &  &  & \{0,-2\} \\
\frac{1}{2} (-k - k_4 - k_5 + k_6 + k_7 - k_8) & 0 & \nep & -1/2 &  \\
\frac{1}{2} (-k - k_4 - k_5 + k_6 - k_7 + k_8) & 0 & \nmp & -1/2 &  \\
&  &  &  &  \\ \hline
&  &  &  &  \\
-\frac{1}{2} (k - k_4 - k_5 + k_6 + k_7 + k_8) & 1 & \tau^+ & 1/2 &  \\
-\frac{1}{2} (k - k_4 - k_5 + k_6 - k_7 - k_8) & 1 & \chi^+ & 1/2 &  \\
&  &  &  & \{0,2\} \\
-\frac{1}{2} (-k - k_4 - k_5 + k_6 + k_7 - k_8) & 0 & \nem & 1/2 &  \\
-\frac{1}{2} (-k - k_4 - k_5 + k_6 - k_7 + k_8) & 0 & \nmm & 1/2 &  \\
&  &  &  &  \\ \hline
&  &  &  &  \\
\frac{1}{2} (k - k_4 - k_5 - k_6 + k_7 - k_8) & -1 & \mu^- & -1/2 &  \\
\frac{1}{2} (k - k_4 - k_5 - k_6 - k_7 + k_8) & -1 & e^- & -1/2 &  \\
\frac{1}{2} (-k - k_4 - k_5 - k_6 + k_7 + k_8) & 0 & \ncp & -1/2 &  \\
\frac{1}{2} (-k - k_4 - k_5 - k_6 - k_7 - k_8) & 0 & \ntp & -1/2 &  \\
&  &  &  & \{0,0\} \\
-\frac{1}{2} (k - k_4 - k_5 - k_6 + k_7 - k_8) & 1 & \mu^+ & 1/2 &  \\
-\frac{1}{2} (k - k_4 - k_5 - k_6 - k_7 + k_8) & 1 & e^+ & 1/2 &  \\
-\frac{1}{2} (-k - k_4 - k_5 - k_6 + k_7 + k_8) & 0 & \ncm & 1/2 &  \\
-\frac{1}{2} (-k - k_4 - k_5 - k_6 - k_7 - k_8) & 0 & \ntm & 1/2 &  \\
&  &  &  &  \\ \hline
\end{array}
}}
\end{equation*}%
\caption{The lepton families with their spin-z.}
\label{t:leptons}
\end{table}

\begin{table}[hp]
\begin{equation*}
{\scriptsize {\
\begin{array}{|c|c|c|c|c|}
\hline
&  &  &  &  \\
roots\ \left[\ k^\prime := k_1-k_2-k_3\ \right] & q_{e.m.} & blue\ quark &
\sigma_z:=\frac12(\as,\rho_1+\rho_2) & \{r_f,s_f\} \\
&  &  &  &  \\ \hline
&  &  &  &  \\
\frac{1}{2} (k^\prime + k_4 - k_5 - k_6 + k_7 + k_8) & {\scriptstyle 1/3} &
\bar b^\prime & -{\scriptstyle 1/2} &  \\
\frac{1}{2} (k^\prime + k_4 - k_5 - k_6 - k_7 - k_8) & {\scriptstyle 1/3} &
\bar B^\prime & -{\scriptstyle 1/2} &  \\
&  &  &  & \{1,1\} \\
\frac{1}{2} (-k^\prime + k_4 - k_5 - k_6 + k_7 - k_8) & {\scriptstyle 2/3} &
u & -{\scriptstyle 1/2} &  \\
\frac{1}{2} (-k^\prime + k_4 - k_5 - k_6 - k_7 + k_8) & {\scriptstyle 2/3} &
c & -{\scriptstyle 1/2} &  \\
&  &  &  &  \\ \hline
&  &  &  &  \\
-\frac{1}{2} (k^\prime + k_4 - k_5 - k_6 + k_7 + k_8) & -{\scriptstyle 1/3}
& b^\prime & {\scriptstyle 1/2} &  \\
-\frac{1}{2} (k^\prime + k_4 - k_5 - k_6 - k_7 - k_8) & -{\scriptstyle 1/3}
& B^\prime & {\scriptstyle 1/2} &  \\
&  &  &  & \{-1,-1\} \\
-\frac{1}{2} (-k^\prime + k_4 - k_5 - k_6 + k_7 - k_8) & -{\scriptstyle 2/3}
& \bar u & {\scriptstyle 1/2} &  \\
-\frac{1}{2} (-k^\prime + k_4 - k_5 - k_6 - k_7 + k_8) & -{\scriptstyle 2/3}
& \bar c & {\scriptstyle 1/2} &  \\
&  &  &  &  \\ \hline
&  &  &  &  \\
\frac{1}{2} (k^\prime - k_4 + k_5 - k_6 + k_7 + k_8) & -{\scriptstyle 2/3} &
\bar t & {\scriptstyle 1/2} &  \\
\frac{1}{2} (k^\prime - k_4 + k_5 - k_6 - k_7 - k_8) & -{\scriptstyle 2/3} &
\bar T & {\scriptstyle 1/2} &  \\
&  &  &  & \{-1,1\} \\
\frac{1}{2} (-k^\prime - k_4 + k_5 - k_6 + k_7 - k_8) & -{\scriptstyle 1/3}
& d^\prime & {\scriptstyle 1/2} &  \\
\frac{1}{2} (-k^\prime - k_4 + k_5 - k_6 - k_7 + k_8) & -{\scriptstyle 1/3}
& s^\prime & {\scriptstyle 1/2} &  \\
&  &  &  &  \\ \hline
&  &  &  &  \\
-\frac{1}{2} (k^\prime - k_4 + k_5 - k_6 + k_7 + k_8) & {\scriptstyle 2/3} &
t & -{\scriptstyle 1/2} &  \\
-\frac{1}{2} (k^\prime - k_4 + k_5 - k_6 - k_7 - k_8) & {\scriptstyle 2/3} &
T & -{\scriptstyle 1/2} &  \\
&  &  &  & \{1,-1\} \\
-\frac{1}{2} (-k^\prime - k_4 + k_5 - k_6 + k_7 - k_8) & {\scriptstyle 1/3}
& \bar d^\prime & -{\scriptstyle 1/2} &  \\
-\frac{1}{2} (-k^\prime - k_4 + k_5 - k_6 - k_7 + k_8) & {\scriptstyle 1/3}
& \bar s^\prime & -{\scriptstyle 1/2} &  \\
&  &  &  &  \\ \hline
&  &  &  &  \\
\frac{1}{2} (k^\prime - k_4 - k_5 + k_6 + k_7 + k_8) & -{\scriptstyle 2/3} &
\bar t & -{\scriptstyle 1/2} &  \\
\frac{1}{2} (k^\prime - k_4 - k_5 + k_6 - k_7 - k_8) & -{\scriptstyle 2/3} &
\bar T & -{\scriptstyle 1/2} &  \\
&  &  &  & \{0,-2\} \\
\frac{1}{2} (-k^\prime - k_4 - k_5 + k_6 + k_7 - k_8) & -{\scriptstyle 1/3}
& d^\prime & -{\scriptstyle 1/2} &  \\
\frac{1}{2} (-k^\prime - k_4 - k_5 + k_6 - k_7 + k_8) & -{\scriptstyle 1/3}
& s^\prime & -{\scriptstyle 1/2} &  \\
&  &  &  &  \\ \hline
&  &  &  &  \\
-\frac{1}{2} (k^\prime - k_4 - k_5 + k_6 + k_7 + k_8) & {\scriptstyle 2/3} &
t & {\scriptstyle 1/2} &  \\
-\frac{1}{2} (k^\prime - k_4 - k_5 + k_6 - k_7 - k_8) & {\scriptstyle 2/3} &
T & {\scriptstyle 1/2} &  \\
&  &  &  & \{0,2\} \\
-\frac{1}{2} (-k^\prime - k_4 - k_5 + k_6 + k_7 - k_8) & {\scriptstyle 1/3}
& \bar d^\prime & {\scriptstyle 1/2} &  \\
-\frac{1}{2} (-k^\prime - k_4 - k_5 + k_6 - k_7 + k_8) & {\scriptstyle 1/3}
& \bar s^\prime & {\scriptstyle 1/2} &  \\
&  &  &  &  \\ \hline
&  &  &  &  \\
\frac{1}{2} (k^\prime - k_4 - k_5 - k_6 + k_7 - k_8) & -{\scriptstyle 2/3} &
\bar c & -{\scriptstyle 1/2} &  \\
\frac{1}{2} (k^\prime - k_4 - k_5 - k_6 - k_7 + k_8) & -{\scriptstyle 2/3} &
\bar u & -{\scriptstyle 1/2} &  \\
\frac{1}{2} (-k^\prime - k_4 - k_5 - k_6 + k_7 + k_8) & -{\scriptstyle 1/3}
& B^\prime & -{\scriptstyle 1/2} &  \\
\frac{1}{2} (-k^\prime - k_4 - k_5 - k_6 - k_7 - k_8) & -{\scriptstyle 1/3}
& b^\prime & -{\scriptstyle 1/2} &  \\
&  &  &  & \{0,0\} \\
-\frac{1}{2} (k^\prime - k_4 - k_5 - k_6 + k_7 - k_8) & {\scriptstyle 2/3} &
c & {\scriptstyle 1/2} &  \\
-\frac{1}{2} (k^\prime - k_4 - k_5 - k_6 - k_7 + k_8) & {\scriptstyle 2/3} &
u & {\scriptstyle 1/2} &  \\
-\frac{1}{2} (-k^\prime - k_4 - k_5 - k_6 + k_7 + k_8) & {\scriptstyle 1/3}
& \bar B^\prime & {\scriptstyle 1/2} &  \\
-\frac{1}{2} (-k^\prime - k_4 - k_5 - k_6 - k_7 - k_8) & {\scriptstyle 1/3}
& \bar b^\prime & {\scriptstyle 1/2} &  \\
&  &  &  &  \\ \hline
\end{array}
}}
\end{equation*}%
\caption{The flavor families of blue (\{1,1\}) quarks with their spin-z. }
\label{t:quarks}
\end{table}

\end{document}